\documentclass[11pt]{article}
\usepackage[utf8]{inputenc}
\usepackage[T1]{fontenc}
\usepackage{amsmath,amsthm,amssymb}
\usepackage{fullpage}
\usepackage{libertine}
\usepackage{booktabs}
\usepackage[ruled]{algorithm2e}
\usepackage{hyperref}
\usepackage{cleveref}
\usepackage{thmtools,thm-restate}
\usepackage[authoryear]{natbib}
\usepackage{braket}
\usepackage{tikz}
\usetikzlibrary{positioning,arrows.meta,shapes,calc,shapes.geometric,intersections}
\usepackage{latex_abbreviations}
\DeclareMathOperator*{\argmax}{argmax}

\newcommand{\E}{\mathop{\mathbb{E}}}

\newcommand{\PoA}{{\mathsf{PoA}}}
\newcommand{\OPT}{{\mathsf{OPT}}}
\newcommand{\STR}{{\mathsf{STR}}}
\newcommand{\PoS}{{\mathsf{PoS}}}

\newcommand{\vSW}{v_{\mathsf{SW}}}
\newcommand{\phiSF}{\phi_{\mathsf{SF}}}

\newtheorem{theorem}{Theorem}[section]
\theoremstyle{plain}
\newtheorem{lemma}[theorem]{Lemma}
\newtheorem{corollary}[theorem]{Corollary}
\newtheorem{proposition}[theorem]{Proposition}

\theoremstyle{definition}
\newtheorem{definition}[theorem]{Definition}

\newtheorem{remark}[theorem]{Remark}
\newtheorem{example}[theorem]{Example}

\title{The power of mediators: Price of anarchy and stability in Bayesian games with submodular social welfare}
\author{Kaito Fujii\footnote{National Institute of Informatics and SOKENDAI. Email: \texttt{fujiik@nii.ac.jp}}}

\begin{document}

\maketitle

\begin{abstract}
This paper investigates the role of mediators in Bayesian games by examining their impact on social welfare through the price of anarchy (PoA) and price of stability (PoS). Mediators can communicate with players to guide them toward equilibria of varying quality, and different communication protocols lead to a variety of equilibrium concepts collectively known as Bayes (coarse) correlated equilibria. To analyze these equilibrium concepts, we consider a general class of Bayesian games with submodular social welfare, which naturally extends valid utility games and their variant, basic utility games. These frameworks, introduced by \citet{Vet}, have been developed to analyze the social welfare guarantees of equilibria in games such as competitive facility location, influence maximization, and other resource allocation problems.

We provide upper and lower bounds on the PoA and PoS for a broad class of Bayes (coarse) correlated equilibria. Central to our analysis is the strategy representability gap, which measures the multiplicative gap between the optimal social welfare achievable with and without knowledge of other players' types. For monotone submodular social welfare functions, we show that this gap is $1-1/\rme$ for independent priors and $\Theta(1/\sqrt{n})$ for correlated priors, where $n$ is the number of players. These bounds directly lead to upper and lower bounds on the PoA and PoS for various equilibrium concepts, while we also derive improved bounds for specific concepts by developing smoothness arguments. Notably, we identify a fundamental gap in the PoA and PoS across different classes of Bayes correlated equilibria, highlighting essential distinctions among these concepts.
\end{abstract}

\tableofcontents

\section{Introduction}
Mediators play a significant role in shaping outcomes in Bayesian games, where players make decisions under uncertainty about other players' private information.
By communicating with players and facilitating coordination, mediators can influence players' decision-making, guiding them toward equilibria with varying quality of social welfare.
Depending on the communication protocols employed, various equilibrium concepts arise, collectively referred to as Bayes (coarse) correlated equilibria \citep{Forges93,Fujii23} (see \Cref{figure}), which have recently attracted attention in information design studies \citep{BM19}.
These equilibrium concepts generalize Nash equilibria by allowing for coordination among players assisted by mediators, opening up the possibilities both improved and degraded social welfare.

A key question in the study of Bayesian games is how mediators affect the quality of equilibria, particularly in terms of the price of anarchy (PoA) and the price of stability (PoS).
The PoA measures the extent to which mediators can deteriorate the quality of equilibria compared to the optimal social welfare, whereas the PoS quantifies how close mediators can bring the best achievable equilibria to the optimal social welfare.
Understanding these metrics is crucial for evaluating the positive and negative influence of mediators and designing communication strategies to optimize social welfare in real-world policymaking.

To address this question, we focus on a general class of Bayesian games with submodular social welfare, which naturally extends valid utility games and their variant, basic utility games in the complete-information setting.
Since the introduction by \citet*{Vet}, valid and basic utility games have stood out as one of the most extensively studied frameworks of resource allocation games with diverse applications, including competitive facility location \citep*{Vet,BT17}, competitive influence maximization \citep*{BKS07,HK13}, and selfish routing \citep*{Vet}.

In the complete-information setting, \citet*{Vet} proved that the PoA for mixed Nash equilibria is guaranteed to be at least $1/2$ in any valid utility game.
This bound was later extended to (coarse) correlated equilibria, a class of equilibria realized by mediators, by developing a general proof technique known as smoothness arguments by \citet*{Roughgarden15}.
In other words, mediators cannot reduce the quality of equilibria to less than half of the optimal social welfare, demonstrating that valid utility games are robust against malicious mediators.

The goal of this study is to analyze the PoA and PoS in the Bayesian version of valid utility games.
We suppose that each player $i \in [n]$ possesses private information represented by a random variable $\theta_i \in \Theta_i$, referred as their \textit{type}, and the type profile $(\theta_1,\theta_2,\dots,\theta_n)$ is generated from a common prior distribution $\rho \in \Delta(\Theta_1 \times \dots \times \Theta_n)$.
To extend valid utility games to Bayesian setting, we consider a varying action set $A_i^{\theta_i}$ depending on their own type $\theta_i$ for each player $i$.\footnote{%
This definition based on type-dependent action sets is different from the standard formulation with fixed action sets $A_i$ and type-dependent utility functions $v_i \colon (\Theta_1 \times \dots \times \Theta_n) \times (A_1 \times \dots \times A_n) \to \bbR$.
However, as described in \Cref{setting}, any Bayesian game can be reformulated with type-dependent action spaces by treating each type-action pair $(\theta_i, a_i)$ as an action. The key feature of our formulation is the existence of a submodular social welfare function consistent across all type profiles, which facilitates our theoretical analysis.
}
If there is a monotone submodular function $f \colon 2^E \to \bbR_{\ge 0}$, where $E = {\bigcup_{i \in N} \bigcup_{\theta_i \in \Theta_i} A_i^{\theta_i}}$, that serves as a social welfare function satisfying the conditions for valid utility games for all type profiles, we refer to this Bayesian game as a \textit{Bayesian valid utility game} (Bayesian basic utility games can be defined similarly).

In practical examples of Bayesian valid and basic utility games, mediators commonly emerge as follows:
\begin{enumerate}
\item Transportation planners (mediators) coordinate commuters (players) with varying schedules (types) to choose transit options and avoid peak-hour congestion.
\item Volunteer coordinators (mediator) assist participants (players) with different availability or capabilities (type) in selecting suitable tasks.
\item Network planners (mediator) help devices (players) with different locations and demands (type) select appropriate base stations.
\end{enumerate}
Submodularity naturally arises in these applications, as the marginal gain in social welfare diminishes due to ``crowding'' when many users select the same resource.
See \Cref{applications} for more detailed descriptions on these applications.

\begin{table}[t]
\caption{Summary of our main results. Note that the PoA decreases as the equilibrium concept is coarsened, while the PoS increases. Each of ``valid'' and ``basic'' stands for valid and basic utility games, and ``independent'' and ``correlated'' represent the assumption for type prior distributions.}
\label{table}
\centering
\begin{tabular}{c|cccccccc}
& BNE, SF/ANFCE, ANF/SFCCE & Com.Eq. & BS, ANF/SFCBS \\
\hline
PoA (valid, independent) & $1/2$ & $1/2$ & $\in \left[ \frac{1-1/\rme}{2}, 0.441 \right]$ \\
PoA (valid, correlated) & $\Theta\left(\frac{1}{\sqrt{n}}\right)$ & $\Theta\left(\frac{1}{\sqrt{n}}\right)$ & $\Theta\left(\frac{1}{\sqrt{n}}\right)$ \\
PoS (basic, independent) & $1-1/\rme$ & $\le 4/5$ & $1$ \\
PoS (basic, correlated) & $\Theta\left(\frac{1}{\sqrt{n}}\right)$ & $\le 4/5$ & $1$
\end{tabular}
\end{table}

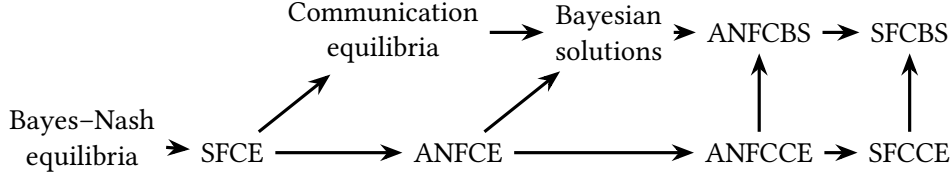
\begin{figure}[t]
\centering
\begin{tikzpicture}[yscale=0.8, line width=1.2pt]
\node [align=center] (BNE) at (0,0.2) {Bayes--Nash\\equilibria};
\node (SFCE) at (2,0) {SFCE};
\node [align=center] (ComEq) at (4,2) {Communication\\equilibria};
\node (ANFCE) at (5,0) {ANFCE};
\node [align=center] (BS) at (7,2) {Bayesian\\solutions};
\node (ANFCCE) at (9,0) {ANFCCE};
\node (SFCCE) at (11,0) {SFCCE};
\node [align=center] (CBS) at (9,2) {ANFCBS};
\node [align=center] (SFCBS) at (11,2) {SFCBS};
\draw [-{Stealth}] (BNE) -- (SFCE);
\draw [-{Stealth}] (SFCE) -- (ANFCE);
\draw [-{Stealth}] (SFCE) -- (ComEq);
\draw [-{Stealth}] (ANFCE) -- (BS);
\draw [-{Stealth}] (ComEq) -- (BS);
\draw [-{Stealth}] (ANFCE) -- (ANFCCE);
\draw [-{Stealth}] (ANFCCE) -- (SFCCE);
\draw [-{Stealth}] (ANFCCE) -- (CBS);
\draw [-{Stealth}] (BS) -- (CBS);
\draw [-{Stealth}] (CBS) -- (SFCBS);
\draw [-{Stealth}] (SFCCE) -- (SFCBS);
\end{tikzpicture}
\caption{Relations of various classes of Bayes correlated equilibria and Bayes coarse correlated equilibria. The tail of each arrow is a subset of its head. Each of ``CE'', ``CCE'', ``CBS'', ``SF'', and ``ANF'' stands for ``correlated equilibria'', ``coarse correlated equilibria'', ``coarse Bayesian solutions'', ``strategic-form'', and ``agent-normal-form'', respectively. See \Cref{bce} for definitions.}
\label{figure}
\end{figure}

\subsection{Our contribution}
In this study, we analyze the PoA and PoS for various equilibrium concepts in Bayesian valid and basic utility games.
The PoA and PoS are defined as the ratio of the expected social welfare value achieved by the worst and best equilibrium to the optimal social welfare value, respectively.
The optimal social welfare value is defined as $\E_{(\theta_1,\dots,\theta_n) \sim \rho} \left[ \max_{(a_1,\dots,a_n) \in A^{\theta_1}_1 \times \dots \times A^{\theta_n}_n} f(\{a_1,\dots,a_n\}) \right]$, selecting an optimal action profile for each type profile.
The optimal action profile can depend on all players' types, while in equilibria, each player knows only their own type.
This difference is unique to Bayesian settings and poses new technical challenges.
To investigate this difference, we introduce a new notion called the \textit{strategy representability gap} (SR gap), which represents the multiplicative difference between the optimal social welfare when each player's action depends only on their own type and when it depends on the entire type profile.

We separately analyze the cases with or without the assumption that players' types are stochastically independent.
In the case of independent prior distributions, as an application of the correlation gap of monotone submodular functions \citep*{ADSY12}, the SR gap is bounded below by $1-1/\rme$ (\Cref{sec:sr-independent}).
The analysis becomes more involved for the $\Theta(1/\sqrt{n})$ lower bound for general correlated prior distributions (\Cref{sec:sr-correlated}).
A critical part of the proof is to analyze at most $\sqrt{n}$ heavy-weight actions and other light-weight actions separately for each player $i \in N$ and type $\theta_i \in \Theta_i$ using the multilinear extension of submodular functions, which is a key tool for submodular maximization \citep*{CCPV11}.
Moreover, these lower bounds on the SR gap are proved to be tight.

Once we establish bounds on the SR gap, we can directly derive upper and lower bounds on the PoA and PoS for various concepts of Bayes (coarse) correlated equilibria by incorporating the smoothness argument.
Additionally, in the case of independent prior distributions, these lower bounds on the PoA can be improved to $1/2$ for some concepts of Bayes correlated equilibria (called communication equilibria and agent-normal-form correlated equilibria), while we provide an upper bound of $0.441$ for another concept (called Bayesian solutions).
This reveals a gap in the PoA among natural extensions of correlated equilibria, highlighting the nuanced differences between these equilibrium concepts.
A similar gap also appears in the PoS for Bayesian basic utility games.
The distinction between the different definitions of Bayes correlated equilibria has been overlooked in most existing studies on the PoA and PoS.
To our knowledge, our result is the first to show that the PoA and PoS can differ among these concepts.

Our main contributions are summarized as follows (see \Cref{table}).
The equilibrium concepts relevant to our results are briefly introduced in \Cref{concepts}, while more detailed descriptions of all concepts shown in \Cref{figure} are provided in \Cref{bce}.
\begin{itemize}
\item In Bayesian valid utility games with independent or correlated prior distributions, the PoA for strategic-form coarse Bayesian solutions is at least $(1-1/\rme)/2$ and $\Omega(1/\sqrt{n})$, respectively (\Cref{subsection-poa-cbs}).
\item In Bayesian valid utility games with independent prior distributions, the PoA for strategic-form coarse correlated equilibria (\Cref{subsection-poa-sfcce}) and communication equilibria (\Cref{subsection-poa-comeq}) is shown to be $1/2$ using extended smoothness arguments.
\item There exists a Bayesian valid utility game with an independent prior distribution in which the PoA for Bayesian solutions is at most $0.441$ (\Cref{subsection-poa-bs}).
\item In Bayesian basic utility games, the PoS for Bayesian solutions is $1$ (\Cref{subsection-pos-bs}).
\item In Bayesian basic utility games with independent or correlated prior distributions, the PoS for Bayes--Nash equilibria is $1-1/\rme$ and $\Theta(1/\sqrt{n})$, respectively (\Cref{subsection-pos-bne}).
\item There exists a Bayesian basic utility game with an independent prior distribution in which the PoS for communication equilibria is $4/5$ (\Cref{subsection-pos-comeq}).
\end{itemize}

\begin{remark}
We can obtain almost the same PoA results for Bayesian basic utility games as those for the valid case, which is a more general setting.
The tight upper bound of $1/2$ established in the non-Bayesian setting by \citet{Vet} also applies to the basic case and extends directly to the Bayesian setting.
Since Bayesian basic utility games are a special case of Bayesian valid utility games, all lower bounds established for the valid case also apply. In particular, in the independent prior case, the PoA lower bound of $1/2$ holds for many equilibrium concepts, including ANFCEs and communication equilibria, and is tight. In the correlated prior case, the lower bound of $\Omega(1/\sqrt{n})$ applies to all solution concepts considered in this paper and is tight up to a constant factor from \Cref{poa-ub-sr-gap}.
The only result specific to the valid case is the $0.441$ upper bound for Bayesian solutions in the independent prior case. Whether this result extends to the basic case remains an open question.
\end{remark}

\subsection{Related work}
\citet*{Vet} first proposed valid and basic utility games with applications to competitive facility location, selfish routing, and multiple-item auctions, and proved that their PoA for mixed Nash equilibria is at least $1/2$.
\citet*{Vet} also proved that an action profile that maximizes the social welfare function is a pure Nash equilibrium in basic utility games, which implies that the PoS for pure Nash equilibria equals $1$.
Since then, the framework of valid (and basic) utility games has been studied as a main tool for analyzing the PoA \citep*{BKS07,HK13}.

\citet*{Roughgarden15} proposed a unified proof framework called smoothness and extended the PoA lower bound from mixed Nash equilibria to coarse correlated equilibria.
The smoothness framework was extended to Bayesian games by \citet*{Roughgarden15,Syr12,ST13} to analyze Bayesian mechanisms, Bayesian congestion games, and other applications.
However, to our knowledge, there has been no investigation into Bayesian valid utility games.

The PoA and PoS in Bayesian games have been explored mainly for Bayes--Nash equilibria, but there are a few notable exceptions that have examined Bayes (coarse) correlated equilibria.
\citet*{HST15} considered the PoA for agent-normal-form coarse correlated equilibria (ANFCCEs) in smooth mechanisms with no-regret dynamics.
\citet{CKK15} analyzed the PoA for a concept of Bayes coarse correlated equilibria, referred to in our terminology as agent-normal-form coarse Bayesian solutions (ANFCBSs), in generalized second-price auctions
Additionally, \citet*{JL23} provided lower bounds on the PoS for Bayesian solutions and ANFCBSs in Bayesian first-price auctions.
\citet*{Fujii23} extended the PoA bound via smoothness to the intersection of communication equilibria and ANFCEs, and also proposed no-regret dynamics converging to it.

Information design is a field that analyzes how mediators (referred to as \textit{sender}) influence the decision-making processes of players (referred to as \textit{receivers}).
Bayes correlated equilibria have gained significant attention as an outcome of information design in multi-receiver settings \citep{BM19}.
Recent studies \citep{AB19,BB17,DX17,FS22} have explored algorithmic information design with submodular maximization.
However, these papers focus on Bayesian persuasion scenarios, where the sender possesses an informational advantage over the receivers and seeks to maximize their own payoff, rather than the social welfare.
In contrast, out study considers a setting where the receivers have private information, and the sender aims to either minimize or maximize social welfare.

\section{Preliminaries}

\subsection{Bayesian games with submodular social welfare}
Let $N = \{1,2,\dots,n\}$ be the set of all players.
The subscript $-i$ denotes all players except player $i \in N$.
Each player $i \in N$ is associated with a type $\theta_i \in \Theta_i$, where $\Theta_i$ is the set of finite possible types.
Let $\Theta = \Theta_1 \times \dots \times \Theta_n$ be the set of all type profiles.
The type profile $\theta = (\theta_1,\dots,\theta_n) \in \Theta$ is randomly generated from a commonly known prior distribution $\rho \in \Delta(\Theta)$, and each player $i \in N$ is notified of their type $\theta_i$.
The prior distribution $\rho$ is called a product distribution or independent if there exist $\rho_i \in \Delta(\Theta_i)$ for each $i \in N$ such that $\rho(\theta) = \prod_{i \in N} \rho_i(\theta_i)$ for all $\theta \in \Theta$.
We write the marginal distribution of $\theta_i$ as $\rho_i$ and the distribution of $\theta_{-i}$ conditioned on $\theta_i \in \Theta_i$ as $\rho|_{\theta_i}$.

In a Bayesian valid or basic utility game, each player $i \in N$ with type $\theta_i \in \Theta_i$ selects one from a set of finite possible actions $A_i^{\theta_i}$.
We can assume that $A_i^{\theta_i}$ is disjoint among all $i \in N$ and $\theta_i \in \Theta_i$ without loss of generality.
Let $A_i = \bigcup_{\theta_i \in \Theta_i} A_i^{\theta_i}$ be the set of all actions of player $i \in N$ and $A = A_1 \times \dots \times A_n$ the set of all action profiles.
The utility function of each player $i \in N$ is defined as $v_i \colon A \to \bbR_{\ge 0}$.
Let $A^\theta = A_1^{\theta_1} \times \dots \times A_n^{\theta_n}$ be the set of all action profiles for the type profile $\theta \in \Theta$, where each player $i \in N$ must choose their action from $A_i^{\theta_i}$.

In Bayesian games, each player's strategy can be represented by a map that associates each type to an action.
Let $S_i = \prod_{\theta_i \in \Theta_i} A_i^{\theta_i}$ be the set of all possible strategies of player $i \in N$.
Each strategy $s_i \in S_i$ determines the action $s_i(\theta_i) \in A_i^{\theta_i}$ of player $i \in N$ for each type $\theta_i \in \Theta_i$.
Let $S = S_1 \times \dots \times S_n$ be the set of all strategy profiles.
For notational simplicity, we use $s(\theta)$ to represent $(s_1(\theta_1),\dots,s_n(\theta_n))$.

Let $E = \bigcup_{i \in N} A_i$ be the set of all actions.
To define the social welfare function $\vSW \colon A \to \bbR_{\ge 0}$, we prepare a set function $f \colon 2^E \to \bbR_{\ge 0}$ with the ground set $E$.
Using this set function, the social welfare function is defined by $\vSW(a) = f(\{a_1,\dots,a_n\})$.
In valid and basic utility games, this set function is assumed to be non-negative, monotone, and submodular.
We define a set function to be monotone if $f(X) \le f(Y)$ for any $X, Y \subseteq E$ with $X \subseteq Y$, and submodular if $f(X \cup \{u\}) - f(X) \ge f(Y \cup \{u\}) - f(Y)$ for any $X, Y \subseteq E$ with $X \subseteq Y$ and $u \in E \setminus Y$.

This varying-action-set formulation can capture not only scenarios in which the type directly represents the set of actions, but also scenarios in which the type represents the weight of each player by replicating actions for each type. 
See \Cref{setting} for more detailed discussions.

\begin{remark}
Note that in the original study \citep*{Vet}, each action $u \in E$ is defined as a subset of the set of resources $\tilde{E}$.
The social welfare function $\tilde{f} \colon 2^{\tilde{E}} \to \bbR_{\ge 0}$ is a set function with the ground set $\tilde{E}$, and the social welfare value is $\tilde{f}(\tilde{a}_1 \cup \dots \cup \tilde{a}_n)$ if each player $i \in N$ chooses $\tilde{a}_i \subseteq \tilde{E}$.
This original notation can be reduced to our definition by setting $f \colon 2^E \to \bbR_{\ge 0}$ as $f(X) = \tilde{f}(\bigcup_{u \in X} u)$ for each $X \subseteq E$, treating each action $u \in E$ as a single element.
These different definitions are equivalent in terms of the PoA and PoS.
\end{remark}

A Bayesian game with the social welfare function $\vSW$ represented by a non-negative monotone submodular function $f$ is called a \textit{Bayesian valid utility game} if
\begin{itemize}
\item[(i)] $\vSW(a) \ge \sum_{i \in N} v_i(a)$ for any $a \in A$, which we call \textit{total utility condition}, and
\item[(ii)] $v_i(a) \ge \vSW(a) - \vSW(\emptyset_i, a_{-i})$ hold for any $a \in A$ and $i \in N$, which we call \textit{marginal contribution condition}, where $\emptyset_i$ represents a (hypothetical) action of player $i \in N$ selecting nothing.
\end{itemize}
If the marginal contribution condition holds with equality for every $a \in A$ and $i \in N$, the game is called a \textit{Bayesian basic utility game}.
A non-Bayesian special case ($\Theta$ is a singleton) is called a valid or basic utility game, respectively.

\subsection{Equilibrium concepts in Bayesian games}
\label{concepts}
We concisely introduce definitions of equilibrium concepts in Bayesian games relevant to our results.
For completeness, we provide a detailed overview of equilibrium concepts in \Cref{bce}.

First, we introduce communication equilibria, Bayesian solutions, and strategic-form coarse Bayesian solutions (SFCBSs), which can be naturally realized by mediators who communicate with players bidirectionally.
These concepts are expressed as a type-dependent action-profile distribution $\pi \in \prod_{\theta \in \Theta} \Delta(A^\theta)$.
Once type profile $\theta \sim \rho$ is generated, each player $i \in N$ privately reports their type $\theta_i$ to the mediator.
Based on the reported type profile $\theta \in \Theta$, the mediator privately recommends $a_i \in A_i$ to each player $i \in N$ according to $a \sim \pi(\theta)$.
The type-dependent distribution $\pi$ is called a communication equilibrium if no player can improve their payoff by reporting an untrue type $\theta' \in \Theta_i$ or choosing a non-recommended action $\phi(a_i) \in A_i^{\theta_i}$.
Bayesian solutions specifically consider the latter deviations, which are natural when the mediator can verify the reported types or knows the type profile in advance.
SFCBSs further assume that each player can deviate unilaterally to a single strategy independent of the recommended action.

\begin{definition}\label{def:com}
A type-dependent distribution $\pi \in \prod_{\theta \in \Theta} \Delta(A^\theta)$ is a communication equilibrium
if
for any $i \in N$, $\theta_i,\theta'_i \in \Theta_i$, and $\phi \colon A_i^{\theta'_i} \to A_i^{\theta_i}$, it holds that
\begin{equation*}
\E_{\theta_{-i} \sim \rho|_{\theta_i}} \left[ \E_{a \sim \pi(\theta)} \left[ v_i(a) \right] \right]
\ge
\E_{\theta_{-i} \sim \rho|_{\theta_i}} \left[ \E_{a \sim \pi(\theta'_i,\theta_{-i})} \left[ v_i(\phi(a_i), a_{-i}) \right] \right].
\end{equation*}
Bayesian solutions are defined by the constraints only for $\theta'_i = \theta_i$.
A type-dependent distribution $\pi \in \prod_{\theta \in \Theta} \Delta(A^\theta)$ is a \textit{strategic-form coarse Bayesian solution} (SFCBS)
if
for any $i \in N$ and $s_i \in S_i$, it holds that
\begin{equation*}
\E_{\theta \sim \rho} \left[ \E_{a \sim \pi(\theta)} \left[ v_i(a) \right] \right]
\ge
\E_{\theta \sim \rho} \left[ \E_{a \sim \pi(\theta)} \left[ v_i(s_i(\theta_i), a_{-i}) \right] \right].
\end{equation*}
\end{definition}

Next, we introduce strategic-form coarse correlated equilibria (SFCCEs) and Bayes--Nash equilibria, which are defined as a strategy-profile distribution $\sigma \in \Delta(S)$.
In an SFCCE, a mediator recommends a strategy $s_i$ to each player $i \in N$ according to $s \sim \sigma$ without knowing the type profile $\theta \sim \rho$.
If no player can improve their payoff by deviating unilaterally to a single strategy $s'_i \in S_i$ independent of the recommended action, $\sigma$ is called an SFCCE.
Bayes--Nash equilibria are a natural generalization of Nash equilibrium to Bayesian games.

\begin{definition}
A distribution $\sigma \in \Delta(S)$ is a strategic-form coarse correlated equilibrium
if
for any $i \in N$ and $s'_i \in S_i$, it holds that
\begin{equation*}
\E_{\theta \sim \rho} \left[ \E_{s \sim \sigma} \left[ v_i(s(\theta)) \right] \right]
\ge                                        
\E_{\theta \sim \rho} \left[ \E_{s \sim \sigma} \left[ v_i(s'_i(\theta_i), s_{-i}(\theta_{-i})) \right] \right].
\end{equation*}
Furthermore, if the recommended strategies $s_1,s_2,\dots,s_n$ are stochastically independent, then $\sigma$ is a Bayes--Nash equilibrium.
\end{definition}

While a type-dependent action-profile distribution $\pi \in \prod_{\theta \in \Theta} \Delta(A^\theta)$ can decide the action profile $a \sim \pi(\theta)$ depending on the type profile $\theta \sim \rho$, a strategy distribution $\sigma \in \Delta(S)$ generate $s \sim \sigma$ independent of $\theta \sim \rho$ and then decides $s(\theta) \in A^\theta$.
The former is strictly broader than the latter, and a type-dependent action-profile distribution $\pi \in \prod_{\theta \in \Theta} \Delta(A^\theta)$ is called \textit{strategy-representable} if it can be expressed as a strategy distribution.

For each of these equilibrium concepts, we can define the price of anarchy and stability as the ratio of the social welfare value achieved by the worst and best equilibrium to that achieved by an optimal action profile for each type profile.

\begin{definition}[{\citep{KP99,ADKTWR08}}]
For a set of equilibria $\Sigma \subseteq \Delta(S)$ defined as distributions over strategy profiles, the price of anarchy and stability are defined as
\begin{align*}
\PoA_\Sigma = \frac{\inf_{\sigma \in \Sigma} \E_{\theta \sim \rho} \left[ \E_{s \sim \sigma} \left[ \vSW(s(\theta)) \right] \right]}{\E_{\theta \sim \rho} \left[ \max_{a \in A^\theta} \vSW(a) \right]}
&&
\text{and}
&&
\PoS_\Sigma = \frac{\sup_{\sigma \in \Sigma} \E_{\theta \sim \rho} \left[ \E_{s \sim \sigma} \left[ \vSW(s(\theta)) \right] \right]}{\E_{\theta \sim \rho} \left[ \max_{a \in A^\theta} \vSW(a) \right]}.
\end{align*}
For equilibrium concepts $\Pi \subseteq \prod_{\theta \in \Theta} \Delta(A^\theta)$ defined as a distribution over actions for each type profile, $\PoA_\Pi$ and $\PoA_\Pi$ are defined by replacing the numerator with $\inf / \sup_{\pi \in \Pi} \E_{\theta \sim \rho} \left[ \E_{a \sim \pi(\theta)} \left[ \vSW(a) \right] \right]$.
\end{definition}

\subsection{Submodular functions}
We introduce several useful notations for submodular functions.
As in the submodular optimization literature, given a set $X$ and an element $u$, we use $X + u$ and $X - u$ as shorthands for $X \cup \{u\}$ and $X \setminus \{u\}$, respectively.
For a set function $f \colon 2^E \to \bbR$, we write $f(u | X) = f(X + u) - f(X)$ and $f(Y | X) = f(X \cup Y) - f(X)$ for any $X, Y \subseteq E$ and $u \in E$.

Submodular functions have various useful properties.
In the proof, we use the following bound of the correlation gap for monotone submodular functions \citep*{Von07}.

\begin{proposition}[{\citep*{Von07}}]
\label{correlation-gap}
Let $f \colon 2^E \to \bbR_{\ge 0}$ be any non-negative monotone submodular function.
Let $\pi \in \Delta(2^E)$ be a distribution over subsets of the ground set $E$ and $\hat{\pi}$ its independent counterpart defined by $\hat{\pi}(X) = \prod_{u \in X} p_u \prod_{u \in E \setminus X} (1-p_u)$ for all $X \subseteq E$, where $p_u = \Pr_{X \sim \pi}(u \in X)$.
Then it holds that $\E_{X \sim \hat{\pi}} \left[ f(X) \right] \ge (1-1/\rme) \E_{X \sim \pi} \left[ f(X) \right]$.
\end{proposition}

Next, we present another lemma concerning a distribution that randomly selects one element from each group, where the ground set is partitioned into disjoint groups.
The proof directly follows from the following result proved by \citet*{QS22}.

\begin{proposition}[{\citep*{QS22}}]
\label{weak-negative-regression}
Suppose that a distribution $\pi \in \Delta(2^E)$ satisfies weak negative regression, that is, for any $u \in E$ and any monotone function $f \colon 2^E \to \bbR$, it holds that $
\E_{X \sim \pi} \left[ f(X - u) \mid u \in X \right]
\le
\E_{X \sim \pi} \left[ f(X - u) \mid u \not\in X \right]
$.
Then for any non-negative monotone submodular function $f \colon 2^E \to \bbR$,
it holds that $\E_{X \sim \pi} \left[ f(X) \right] \ge \E_{X \sim \hat{\pi}} \left[ f(X) \right]$,
where $\hat{\pi}$ is the distribution that independently contains each element $u \in E$ with probability $\Pr_{X \sim \pi}(u \in X)$.
\end{proposition}

\begin{lemma}
\label{nc-gap}
Let $f \colon 2^E \to \bbR_{\ge 0}$ be any non-negative monotone submodular function and the ground set $E$ is decomposed into disjoint subsets, i.e., $E = E_1 \cup \dots \cup E_k$.
Each element $u \in E$ is associated with probability $p_u \in [0,1]$ such that $\sum_{u \in E_\ell} p_u = 1$ for each $\ell \in [k]$.
Let $\pi \in \Delta(2^E)$ be the product distribution defined by $\pi(X) = \prod_{\ell=1}^k p_{u_\ell}$ if $X$ contains exactly one element $u_\ell$ from each subset $E_\ell$ and $\pi(X) = 0$ otherwise.
Let $\hat{\pi} \in \Delta(2^E)$ be the distribution that independently contains each element $u \in E$ with probability $p_u$.
Then $\E_{X \sim \pi} \left[ f(X) \right] \ge \E_{X \sim \hat{\pi}} \left[ f(X) \right]$.
\end{lemma}

\begin{proof}
Conditioning $u \in X$, the probability that $X$ includes each other element in the same partition decreases to $0$, while conditioning $u \not\in X$, the probability increases.
This conditioning does not influence the distribution of elements in the other partitions.
Therefore, the distribution $\pi$ satisfies weak negative regression.
From \Cref{weak-negative-regression}, the statement follows.
\end{proof}

For a set function $f \colon 2^E \to \bbR$, its multilinear extension $F \colon [0,1]^E \to \bbR$ is defined as $F(x) = \E_{X \sim x} \left[ f(X) \right]$ for each $x \in [0,1]^E$, where $X \sim x$ represents that $X$ independently contains each element $u \in E$ with probability $x_u$.
The multilinear extension of submodular functions has various useful properties.
Here we introduce some of these properties that will be used in our proofs.

\begin{lemma}
\label{mle}
(i) The derivative of the multilinear extension is $\nabla_u F(x) = \E_{X \sim x} \left[ f(X+u) - f(X-u) \right]$.
(ii) If $f$ is submodular, for any $x \in [0,1]^E$, a function $g \colon [0,1] \to \bbR$ defined by $g(t) = F(tx)$ is concave.
(iii) If $f$ is non-negative and submodular, then $F(x) \le k F(x / k)$ for any $x \in [0,1]^E$ and $k \ge 1$.
\end{lemma}

\begin{proof}
\begin{itemize}
\item[(i)] The derivative can be obtained by differentiating the explicit formula of the multilinear extension
$F(x) = \sum_{X \subseteq E} f(X) \prod_{u \in X} x_u \prod_{u \in E \setminus X} (1-x_u)$.

\item[(ii)] We can check concavity by considering the second derivative $g''(t) = x^\top \nabla^2 F(tx) x $.
From submodularity, every entry of the Hessian of the multilinear extension is non-positive.
Since $x$ is a non-negative vector, $g''(t)$ is always non-positive.
Hence, $g$ is a concave function.

\item[(iii)] Define $g(t) = F(tx)$.
From (ii), $g$ is concave. 
From concavity and non-negativity, $g(x/k) \ge \left( 1- \frac{1}{k} \right) g(0) + \frac{1}{k} g(x) \ge \frac{1}{k} g(x)$, which proves the inequality.
\end{itemize}
\end{proof}

\section{Strategy representability gap}
\label{section-srgap}

In Bayesian games, even if all players' strategies are controlled by the centralized decision-maker, it is impossible to achieve the optimal social welfare value.
We call the multiplicative gap between these two social welfare values the \textit{strategy representability gap} (SR gap for short).

\begin{definition}
We define the strategy representability gap of a Bayesian game as
\begin{equation*}
\frac{
\max_{s \in S} \E_{\theta \sim \rho} \left[ \vSW(s(\theta)) \right]
}{
\E_{\theta \sim \rho} \left[ \max_{a \in A^\theta} \vSW(a) \right]
},
\end{equation*}
where $A^\theta = A_1^{\theta_1} \times \dots \times A_n^{\theta_n}$ is the set of all action profiles under type profile $\theta \in \Theta$.
For a class of Bayesian games, its strategy representability gap is defined as the infimum of the strategy representability gaps across all Bayesian games in this class.
\end{definition}

Note that the SR gap is independent of the distinction between Bayesian valid and basic utility games, as it depends solely on the social welfare function and not on the individual utility functions.

\subsection{Independent case}\label{sec:sr-independent}

First, we provide tight upper and lower bounds on the SR gap for the case where the type prior distribution $\rho$ is a product distribution.

\begin{proposition}
\label{sr-independent}
For any Bayesian valid or basic utility game, the strategy representability gap is at least $1-1/\rme$ if the prior distribution is a product distribution.
\end{proposition}

\begin{proof}
Let $a^\theta \in A^\theta$ be an optimal action profile for each $\theta \in \Theta$, i.e., $a^\theta \in \argmax_{a \in A^\theta} \vSW(a)$.
We consider a distribution $\pi \in \Delta(2^E)$ defined by $\pi(X) = \Pr_{\theta \sim \rho}(X = \{a_i^\theta,\dots,a_n^\theta\})$ for each $X \subseteq E$.
We define its independent counterpart $\hat{\pi} \in \Delta(2^E)$ that independently contains each element $u \in E$ with probability $\Pr_{X \sim \pi}(u \in X)$.
Then from \Cref{correlation-gap}, we have
\begin{equation}
\label{ineq-pihat}
\E_{X \sim \hat{\pi}} \left[ f(X) \right]
\ge (1-1/\rme) \E_{X \sim \pi} \left[ f(X) \right]
= (1-1/\rme) \E_{\theta \sim \rho} \left[ \vSW(a^\theta) \right].
\end{equation}

Let $\pi^* \in \Delta(2^E)$ be the probability distribution that first generates $\theta \in \Theta$ according to $\rho$ and then contains exactly one element $a_i \in A_i^{\theta_i}$ with probability $\Pr_{\theta' \sim \rho}(a^{\theta'}_i = a_i \mid \theta'_i = \theta_i )$ independently for each $i \in N$.
By considering the partition $E = \bigcup_{i \in N} A_i$, we can see that $\pi^*$ and $\hat{\pi}$ satisfy the assumption of \Cref{nc-gap}.
Here we require the independence of $\rho$.
We thus obtain
\begin{equation}
\label{ineq-pistar}
\E_{X \sim \pi^*} \left[ f(X) \right]
\ge 
\E_{X \sim \hat{\pi}} \left[ f(X) \right].
\end{equation}
We can interpret $\pi^*$ as a randomized strategy profile in which each player randomly chooses $a_i \in A_i^{\theta_i}$ with probability $\Pr_{\theta' \sim \rho}(a^{\theta'}_i = a_i \mid \theta'_i = \theta_i )$ when their type is $\theta_i \in \Theta_i$, and then the left-hand side is smaller than the expected social welfare value achieved by an optimal strategy profile.
Therefore, combining \eqref{ineq-pihat} and \eqref{ineq-pistar}, we obtain
\begin{equation*}
\max_{s \in S} \E_{\theta \sim \rho} \left[ \vSW(s(\theta)) \right]
\ge
(1-1/\rme) \E_{\theta \sim \rho} \left[ \max_{a \in A^\theta} \vSW(a) \right]. \qedhere
\end{equation*}

\end{proof}

Moreover, there exists an example in which this lower bound on the SR gap is tight.

\begin{proposition}
\label{sr-gap-independent-ub}
There exists a Bayesian valid utility game with an independent prior distribution in which the strategy representability gap is $1-1/\rme$.
\end{proposition}

\begin{proof}
We consider an example of a coverage function with universe $U$.
Suppose $|U| = n$.
Each player's action space $A_i^{\theta_i}$ is a random subset that contains each element in $U$ with probability $2 \log n / n$.
The social welfare function counts the number of elements in the universe selected by at least one player.

It is known that a Erd\"o{}s--Renyi random bipartite graph with $n$ left and right vertices has a perfect matching with probability converging to $1$ if the edge probability is $\frac{\log n + \omega(1)}{n}$ (see, e.g., Theorem 6.1 of \citep*{FK15}).
Therefore, an optimal allocation of size $n$ from the players to the universe exists with probability approaching to $1$, and the optimal social welfare also converges to $n$.

On the other hand, we can show that the social welfare value achieved by an optimal strategy profile is approximately $(1-1/\rme)n$.
Recall that a strategy profile chooses each player's action depending only on their own type.
Since the type distribution is independent, the choices of the players also follow a product distribution.
Let $x_i \in \Delta([n])$ be the distribution of the action selected by player $i \in N$.
Then $x_i(u) \le \frac{2 \log n}{n}$ holds because $A_i^{\theta_i}$ contains $u$ with probability $\frac{2 \log n}{n}$.
Since $1-x \ge \exp(-\frac{x}{1-x})$ for every $x \in [0,1)$,
the expected number of covered elements is evaluated as
\begin{align*}
\sum_{u \in U} \left( 1-\prod_{i \in N} (1-x_i(u)) \right)
&\le
\sum_{u \in U} \left( 1-\prod_{i \in N} \exp\left(-\frac{x_i(u)}{1-x_i(u)} \right) \right)\\
&=
\sum_{u \in U} \left( 1-\exp\left(- \sum_{i \in N} \frac{x_i(u)}{1-x_i(u)} \right) \right)\\
&\le
n-\sum_{u \in U} \exp\left(- \frac{\sum_{i \in N} x_i(u)}{1-2\log n / n} \right)  \tag{since $x_i(u) \le 2 \log n / n$}\\
&\le
n-n \exp\left(- \frac{1}{n}\sum_{u \in U} \frac{\sum_{i \in N} x_i(u)}{1-2\log n / n} \right) \tag{Jensen's inequality}\\
&\le
n-n \exp\left(- \frac{1}{1-2\log n / n} \right). \tag{since $\sum_{u \in U} x_i(u) = 1$}
\end{align*}
Its ratio to the optimal value $n$ converges to $1-1/\rme$ as $n \to \infty$.
\end{proof}

\subsection{Correlated case}\label{sec:sr-correlated}

Here we establish a lower bound for the SR gap in the correlated case and present a tight example with the SR gap of the same order.
While the SR gap is a constant factor in the independent case, it deteriorates to $\Theta(1/\sqrt{n})$ in the correlated case.
For clarity and simplicity, we assume $n$ is a square number; this assumption can be removed at the cost of a slight decrease in the multiplicative constant.

\begin{theorem}
\label{sr-correlated}
Assume that $\sqrt{n}$ is an integer.
For any Bayesian valid or basic utility game, the strategy representability gap is at least $\displaystyle \frac{1}{(2+ \frac{1}{1-1/\rme}) \sqrt{n}}$.
\end{theorem}

First, we introduce notations used in the proof.
Let $a^\theta \in \argmax_{a \in A^\theta} \vSW(a)$ be an optimal action profile for each type profile $\theta \in \Theta$.
Let $\OPT = \E_{\theta \sim \rho} \left[ \vSW(a^\theta) \right]$ be the optimal social welfare.
Let $w_i^{\theta_i}(a_i) = \Pr_{\theta_{-i} \sim \rho|_{\theta_i}}(a_i^\theta = a_i)$ be the probability that player $i \in N$ with type $\theta_i$ chooses action $a_i \in A_i^{\theta_i}$ in the optimal action profile.
Let $\STR = \max_{s \in S} \E_{\theta \sim \rho} \left[ \vSW(s(\theta)) \right]$ be the social welfare achieved by an optimal strategy profile.
For each $i \in N$ and $\theta_i \in \Theta_i$, we define a vector $y_i^{\theta_i} \in [0,1]^E$ by $y_i^{\theta_i}(e) = \sqrt{n} w_i^{\theta_i}(e)$ if $e \in A_i^{\theta_i}$ and $w_i^{\theta_i}(e) \le 1/\sqrt{n}$, and $y_i^{\theta_i}(e) = 0$ otherwise.
We prove the following lemma with these notations.

\begin{lemma}
\label{lemma-str}
Let $F \colon [0,1]^E \to \bbR_{\ge 0}$ be the multilinear extension of the social welfare function $f \colon 2^E \to \bbR_{\ge 0}$.
Then we have
\begin{equation*}
\STR
\ge
\frac{1-1/\rme}{n} \sum_{i \in N} \E_{\theta \sim \rho} \left[ F(y_i^{\theta_i}) \right].
\end{equation*}
\end{lemma}

\begin{proof}
We consider the randomized strategy profile $\sigma \in \Delta(S)$ that chooses each $a_i \in A_i^{\theta_i}$ with probability $w_i^{\theta_i}(a_i)$ independently for each $i \in N$ and $\theta_i \in \Theta_i$.
Since the value achieved by $\sigma$ is a convex combination of values achieved by strategy profiles, it can be bounded above by the value achieved by an optimal strategy profile as
\begin{equation*}
\STR
=
\max_{s \in S} \E_{\theta \sim \rho} \left[ \vSW(s(\theta)) \right]
\ge
\E_{\theta \sim \rho} \left[ \E_{s \sim \sigma} \left[ \vSW(s(\theta)) \right] \right]
=
\E_{\theta \sim \rho} \left[ \E_{s \sim \sigma} \left[ f(\{s_1(\theta_1),\dots,s_n(\theta_n)\}) \right] \right],
\end{equation*}
where the second equality is due to the definition of $\vSW$.

In the following, we fix any $\theta \in \Theta$ and prove the lower bound $\E_{s \sim \sigma} \left[ f(\{s_1(\theta_1),\dots,s_n(\theta_n)\}) \right] \ge \frac{1-1/\rme}{n} \sum_{i \in N} F(y_i^{\theta_i})$.

From the definition of $\sigma$, an action $s_i(\theta_i)$ selected by player $i \in N$ with type $\theta_i \in \Theta_i$ follows the distribution $w_i^{\theta_i} \in \Delta(A_i^{\theta_i})$.
Let $\hat{\pi} \in \Delta(E)$ be its independent counterpart, that is, each element $a_i \in A_i^{\theta_i}$ for each $i \in N$ is independently contained in $X \sim \hat{\pi}$ with probability $w_i^{\theta_i}(a_i)$.
Since $s_i(\theta_i)$ chooses one element from $A_i^{\theta_i}$ according to the distribution $w_i^{\theta_i}$ for each $i \in N$, \Cref{nc-gap} implies
\begin{equation}
\label{sr-gap-ineq1}
\E_{s \sim \sigma} \left[ f(\{s_1(\theta_1),\dots,s_n(\theta_n)\}) \right]
\ge
\E_{X \sim \hat{\pi}} \left[ f(X) \right].
\end{equation}

Let $N = N_1 \cup \dots \cup N_{\sqrt{n}}$ be an arbitrary partition of $N$ such that $|N_\ell| = \sqrt{n}$ for each $\ell \in [\sqrt{n}]$.
We consider a distribution $\pi_\ell \in \Delta(2^E)$ that contains each element $a_i \in A_i^{\theta_i}$ independently with probability $y_i^{\theta_i}(a_i)$ for each $i \in N_\ell$.
Let $\pi \in \Delta(2^E)$ be the uniform mixture of $\pi_1,\dots,\pi_{\sqrt{n}}$.
For each $i \in N$, the marginal probability that $X \sim \pi$ includes each element $a_i \in A_i^\theta$ is $\frac{1}{\sqrt{n}} \cdot y_i^{\theta_i}(a_i)$, which is at most $w_i^{\theta_i}(a_i)$ from the definition of $y_i^{\theta_i}$.
Let $\hat{\pi}' \in \Delta(2^E)$ be its independent counterpart, that is, each element $u \in E$ is independently included in $X \sim \hat{\pi}'$ with probability $\Pr_{X \sim \pi}(u \in X)$.
Then we obtain
\begin{equation}
\label{sr-gap-ineq2}
\E_{X \sim \hat{\pi}} \left[ f(X) \right]
\ge
\E_{X \sim \hat{\pi}'} \left[ f(X) \right]
\ge
(1-1/\rme) \E_{X \sim \pi} \left[ f(X) \right]
=
\frac{1-1/\rme}{\sqrt{n}} \sum_{\ell=1}^{\sqrt{n}} \E_{X \sim \pi_\ell} \left[ f(X) \right],
\end{equation}
where the first inequality is due to monotonicity of $f$ and $\frac{1}{\sqrt{n}} \cdot y_i^{\theta_i}(a_i) \le w_i^{\theta_i}(a_i)$ for each $i \in N$ and $a_i \in A_i^{\theta_i}$, and the second inequality is due to the correlation gap bound of monotone submodular functions (\Cref{correlation-gap}), and the equality holds because $\pi$ is the uniform mixture of $\pi_1,\dots,\pi_{\sqrt{n}}$.

By comparing $X \sim \pi_\ell$ with its subset $X \cap A_i^{\theta_i}$, where $i \in N_\ell$ is chosen uniformly at random, we can observe from the monotonicity of $f$ that
\begin{equation}
\label{sr-gap-ineq3}
\E_{X \sim \pi_\ell} \left[ f(X) \right]
\ge
 \frac{1}{\sqrt{n}} \sum_{i \in N_\ell} \E_{X \sim \pi_\ell} \left[ f(X \cap A_i^{\theta_i}) \right].
\end{equation}
For each $i \in N_\ell$, by the definition of $\pi_\ell$, when $X \sim \pi_\ell$, the subset $X \cap A_i^{\theta_i}$ independently contains each $a_i \in A_i^{\theta_i}$ with probability $y_i^{\theta_i}(a_i)$.
Then we can write 
\begin{equation}
\label{sr-gap-ineq4}
\E_{X \sim \pi_\ell} \left[ f(X \cap A_i^{\theta_i}) \right] = F(y_i^{\theta_i}),
\end{equation}
where $F \colon [0,1]^E \to \bbR_{\ge 0}$ is the multilinear extension of $f$.

Finally, by combining \eqref{sr-gap-ineq1}, \eqref{sr-gap-ineq2}, \eqref{sr-gap-ineq3}, and \eqref{sr-gap-ineq4}, we obtain
\begin{equation*}
\E_{s \sim \sigma} \left[ f(\{s_1(\theta_1),\dots,s_n(\theta_n)\}) \right]
\ge
\frac{1-1/\rme}{n} \sum_{i \in N} F(y_i^{\theta_i}),
\end{equation*}
which proves the statement.
\end{proof}

\begin{proof}[Proof of \Cref{sr-correlated}]
For each $i \in N$ and $\theta_i \in \Theta_i$, we define a random subset $B_i^{\theta_i} \subseteq A_i^{\theta_i}$ that independently contains each $a_i \in A_i^{\theta_i}$ with probability $y_i^{\theta_i}(a_i)$, and a subset $C_i^{\theta_i} = \{ a_i \in A_i^{\theta_i} \mid w_i^{\theta_i}(a_i) \ge 1/\sqrt{n} \}$ of elements that frequently appear in the optimal action profile $a_i^\theta$.
Since $\sum_{a_i \in A_i^{\theta_i}} w_i^{\theta_i}(a_i) = 1$, we have $|C_i^{\theta_i}| \le \sqrt{n}$.
For each fixed $B_i^{\theta_i}$, it holds that
\begin{align*}
&\sum_{\theta \in \Theta} \rho(\theta) \sum_{i \in N} f(a^{\theta}_i | B^{\theta_i}_i \cup C^{\theta_i}_i)\\
&\ge
\sum_{\theta \in \Theta} \rho(\theta) \sum_{i \in N} f\left( a^{\theta}_i ~\middle|~ \bigcup_{i \in N} (B^{\theta_i}_i \cup C^{\theta_i}_i) \right) \tag{from submodularity}\\
&\ge
\sum_{\theta \in \Theta} \rho(\theta) f\left( \left\{ a_1^{\theta},\dots,a_n^{\theta} \right\} ~\middle|~ \bigcup_{i \in N} (B^{\theta_i}_i \cup C^{\theta_i}_i) \right) \tag{from submodularity}\\
&\ge
\sum_{\theta \in \Theta} \rho(\theta) \left\{ f\left( \left\{ a_1^{\theta},\dots,a_n^{\theta} \right\} \right) - f\left( \bigcup_{i \in N} (B^{\theta_i}_i \cup C^{\theta_i}_i) \right) \right\} \tag{from monotonicity}\\
&=
\OPT - \sum_{\theta \in \Theta} \rho(\theta) f\left( \bigcup_{i \in N} (B^{\theta_i}_i \cup C^{\theta_i}_i) \right)\\
&\ge
\OPT - 
\sum_{\theta \in \Theta} \rho(\theta) f\left( \bigcup_{i \in N} B^{\theta_i}_i \right)
- \sum_{\theta \in \Theta} \rho(\theta) f\left( \bigcup_{i \in N} C^{\theta_i}_i \right),
\end{align*}
where the last inequality is due to submodularity and non-negativity of $f$.
To derive an upper bound for $\OPT$, we separately bound the left-hand side and the negative terms on the right-hand side.

The left-hand side can be expressed as
\begin{align*}
\sum_{\theta \in \Theta} \rho(\theta) \sum_{i \in N} f(a^{\theta}_i | B^{\theta_i}_i \cup C^{\theta_i}_i)
&=
\sum_{i \in N} \sum_{\theta_i \in \Theta_i} \sum_{\theta_{-i} \in \Theta_{-i}} \rho(\theta) f(a^{\theta}_i | B^{\theta_i}_i \cup C^{\theta_i}_i)\\
&=
\sum_{i \in N} \sum_{\theta_i \in \Theta_i} \rho_i(\theta_i) \sum_{a_i \in A_i^{\theta_i}} w_i^{\theta_i}(a_i) f(a_i | B^{\theta_i}_i \cup C^{\theta_i}_i)
\end{align*}
from the definition of $w_i^{\theta_i}$.

Fix any $i \in N$ and $\theta_i \in \Theta_i$.
For each $a_i \in C_i^{\theta_i}$, we have $f(a_i | B^{\theta_i}_i \cup C^{\theta_i}_i) = 0$ and $y_i^{\theta_i}(a_i) = 0$.
For the remaining elements $a_i \in A_i^{\theta_i} \setminus C_i^{\theta_i}$, we have $w_i^{\theta_i}(a_i) = y_i^{\theta_i}(a_i) / \sqrt{n}$.
We thus obtain
\begin{align*}
\E_{B_i^{\theta_i}} \left[ \sum_{a_i \in A_i^{\theta_i}} w_i^{\theta_i}(a_i) f(a_i | B^{\theta_i}_i \cup C^{\theta_i}_i) \right]
&\le
\frac{1}{\sqrt{n}} \sum_{a_i \in A_i^{\theta_i}} y_i^{\theta_i}(a_i) \E_{B_i^{\theta_i}} \left[ f(a_i | B^{\theta_i}_i) \right] \tag{from submodularity}\\
&\le
\frac{1}{\sqrt{n}} \sum_{a_i \in A_i^{\theta_i}} y_i^{\theta_i}(a_i) \E_{B_i^{\theta_i}} \left[ f(B^{\theta_i}_i + a_i) - f(B^{\theta_i}_i - a_i) \right] \tag{from motonocity} \\
&\le
\frac{1}{\sqrt{n}} \sum_{a_i \in A_i^{\theta_i}} y_i^{\theta_i}(a_i) \nabla_{a_i} F(y_i^{\theta_i}),
\end{align*}
from \Cref{mle} (i).
For any non-negative direction, the multilinear extension is a non-decreasing concave function from \Cref{mle} (ii).
Hence, $g \colon [0,1] \to \bbR_{\ge 0}$ defined by $g(t) = F(t y_i^{\theta_i})$ is a non-decreasing concave function.
Since the derivative is non-increasing, we have $g(1) = g(0) + \int_0^1 g'(t) \mathrm{d}t \ge \int_0^1 g'(1) \mathrm{d}t = g'(1)$.
It follows that $\langle y_i^{\theta_i}, \nabla F(y_i^{\theta_i}) \rangle = g'(1) \le g(1) = F(y_i^{\theta_i})$.
Therefore, we obtain
\begin{align*}
\E_{B_i^{\theta_i}} \left[ \sum_{\theta \in \Theta} \rho(\theta) \sum_{i \in N} f(a^{\theta}_i | (B^{\theta_i}_i \cup C^{\theta_i}_i)) \right]
&\le
\frac{1}{\sqrt{n}} \sum_{i \in N} \sum_{\theta_i \in \Theta_i} \rho_i(\theta_i) F(y_i^{\theta_i})
=
\frac{1}{\sqrt{n}} \sum_{i \in N} \sum_{\theta \in \Theta} \rho(\theta) F(y_i^{\theta_i}).
\end{align*}
By applying \Cref{lemma-str}, we can upper-bound the right-hand side by $\sqrt{n} \STR / (1-1/\rme)$.

Next, we bound the expected value of
$
\sum_{\theta \in \Theta} \rho(\theta) f\left( \bigcup_{i \in N} B^{\theta_i}_i \right)
$.
Recall $B_i^{\theta_i}$ is a random set that independently contains each $a_i \in A_i^{\theta_i}$ with probability $y_i^{\theta_i}(a_i)$.
Then we obtain
\begin{equation*}
\sum_{\theta \in \Theta} \rho(\theta) \E_{B_i^{\theta_i}} \left[ f\left( \bigcup_{i \in N} B^{\theta_i}_i \right) \right]
=
\E_{\theta \sim \rho} \left[ \E_{B_i^{\theta_i}} \left[ f\left( \bigcup_{i \in N} B^{\theta_i}_i \right) \right] \right]
=
\E_{\theta \sim \rho} \left[ F\left(\sum_{i \in N} y_i^{\theta_i} \right) \right]
\le
\sqrt{n} \E_{\theta \sim \rho} \left[ F\left(\sum_{i \in N} y_i^{\theta_i} / \sqrt{n} \right) \right],
\end{equation*}
where the inequality is due to \Cref{mle} (iii).
From the definition, it holds that $\sum_{a_i \in A_i^{\theta_i}} y_i^{\theta_i}(a_i) \le \sqrt{n} \sum_{a_i \in A_i^{\theta_i}} w_i^{\theta_i}(a_i) = \sqrt{n}$.
Then the sum of $y_i^{\theta_i} / \sqrt{n}$ is at most $1$.
Therefore, we can consider a distribution that for each $i \in N$ contains at most one element from $A_i^{\theta_i}$ according to $y_i^{\theta_i} / \sqrt{n}$.
By considering a randomized strategy profile that chooses $s_i(\theta_i)$ according to the distribution $y_i^{\theta_i} / \sqrt{n}$ independently for each $i \in N$ and $\theta_i \in \Theta_i$, from \Cref{nc-gap}, we obtain
\begin{equation*}
\sqrt{n} \E_{\theta \sim \rho} \left[ F\left(\sum_{i \in N} y_i^{\theta_i} / \sqrt{n} \right) \right]
\le \sqrt{n} \STR.
\end{equation*}

Next, we bound
$
\sum_{\theta \in \Theta} \rho(\theta) f\left( \bigcup_{i \in N} C^{\theta_i}_i \right)
$.
For each $i \in N$ and $\theta_i \in \Theta_i$, since $|C_i^{\theta_i}| \le \sqrt{n}$, we can give an arbitrary ordering $C_i^{\theta_i} = \{c_i^{\theta_i}(1), \dots, c_i^{\theta_i}(\sqrt{n})\}$.
If $|C_i^{\theta_i}|$ is smaller than $\sqrt{n}$, the same element can appear multiple times in this ordering.
From submodularity, we obtain $f\left( \bigcup_{i \in N} C^{\theta_i}_i \right) \le \sum_{j=1}^{\sqrt{n}} f\left( \left\{c_i^{\theta_i}(j) \mid i \in N \right\} \right)$, and then
\begin{equation*}
\sum_{\theta \in \Theta} \rho(\theta) f\left( \bigcup_{i \in N} C^{\theta_i}_i \right)
=
\E_{\theta \sim \rho} \left[ f\left( \bigcup_{i \in N} C^{\theta_i}_i \right) \right]
\le
\sum_{j=1}^{\sqrt{n}} \E_{\theta \sim \rho} \left[ f\left( \left\{c_i^{\theta_i}(j) \mid i \in N \right\} \right) \right]
\le \sqrt{n} \STR,
\end{equation*}
where the last inequality holds because $s \in S$ defined by $s(\theta) = (c_1^{\theta_1}(j),\dots,c_n^{\theta_n}(j))$ for each $\theta \in \Theta$ can be interpreted as a strategy profile.

Combining all these inequalities, we obtain
\begin{equation*}
\OPT \le \left( 2 + \frac{1}{1-1/\rme} \right) \sqrt{n} \STR. \qedhere
\end{equation*}
\end{proof}

We provide an example that proves the order of this lower bound is tight.

\begin{proposition}
There exists a Bayesian valid utility game in which the strategy representability gap is at most $2 / \sqrt{n}$.
\end{proposition}

\begin{proof}
Suppose that the set of types is indexed by $k$-dimensional integral vectors.
Formally, we define $\Theta_i = L$ for each $i \in N$, where $L = [n]^k$.
We define the prior distribution $\rho$ as follows.
First, we sample $j \in [k]$ and $\ell_1,\ell_2,\dots,\ell_k \in [n]$ uniformly at random.
Then, we assign $n$ types $\{ (\ell_1,\dots,\ell_{j-1},i',\ell_{j+1},\dots,\ell_k) \mid i' \in [n] \}$ to $n$ players using a uniformly random permutation.

Each player with type $\ell \in L$ has $k$ actions, denoted by $A_i^\ell = \{v_1^\ell,\dots,v_k^\ell\}$.
The social welfare function is defined as a coverage function over the universe $U$, where $U = [k] \times [n]$.
The $h$th action $v_h^\ell$ covers $(h, \ell_h) \in U$.
Since the players' types are almost identical except for the $j$th entry, if multiple players choose an action with the same index other than $j$, their total contribution to the social welfare is only $1$.

The optimal social welfare is $n$, which is achieved when every player always chooses the $j$th action.
On the other hand, in an optimal strategy profile, each player does not know the other players' types and therefore cannot always choose the $j$th action.
Each player with type $\ell \in [n]^k$ appears in $k$ type profiles, each corresponding to different $j \in [k]$, but they can play the $j$th action only in one of these $k$ type profiles.
Hence, the expected social welfare obtained by each player's $j$th action is at most $1/k$, and the total contribution from all players' $j$th actions is at most $n/k$.
For each type profile, the contribution to the social welfare from actions other than the $j$th action is at most $k$.
Thus, the social welfare value achieved by an optimal strategy profile is at most $n/k + k$.
In the case of $k = \sqrt{n}$, the SR gap is at most $2 / \sqrt{n}$.
\end{proof}

\section{Price of anarchy in Bayesian valid utility games}

This section provides lower and upper bounds on the PoA in Bayesian valid utility games for various concepts of Bayes (coarse) correlated equilibria.
In \Cref{subsection-poa-cbs}, we derive lower bounds on the PoA for SFCBSs by combining the SR gap bounds and the smoothness arguments: $(1-1/\rme)/2$ in the independent case and $\Theta(1/\sqrt{n})$ in the correlated case.
In \Cref{subsection-poa-sfcce} and \Cref{subsection-poa-comeq}, we improve the lower bounds in the independent case to $1/2$ for SFCCEs and communication equilibria, respectively.
In \Cref{subsection-poa-bs}, we establish an upper bound of $0.441$ for Bayesian solutions when the prior distribution is independent.

\subsection{For strategic-form coarse Bayesian solutions}
\label{subsection-poa-cbs}
To prove the PoA lower bounds, we prepare the following lemma that compares an equilibrium with an optimal strategy profile.
The proof is a straightforward application of the smoothness argument and is deferred to \Cref{proofs}.

\begin{restatable}{lemma}{poacbs}
\label{poa-cbs}
In any Bayesian valid utility game, any strategic-form coarse Bayesian solution $\pi \in \prod_{\theta \in \Theta} \Delta(A^\theta)$ satisfies 
\begin{equation*}
\E_{\theta \sim \rho} \left[ \E_{a \sim \pi(\theta)} \left[ \vSW(a) \right] \right] 
\ge
\frac{1}{2}
\max_{s \in S} \E_{\theta \sim \rho} \left[ \vSW(s(\theta)) \right]
\end{equation*}
\end{restatable}

By combining this lemma with the lower bounds on the strategy representability gap (\Cref{sr-independent} and \Cref{sr-correlated}), we can immediately obtain the PoA lower bounds.

\begin{corollary}
For Bayesian valid utility games with $n$ players, the price of anarchy for SFCBSs is at least $\frac{1-1/\rme}{2}$ and $\displaystyle \Omega(1/\sqrt{n})$ in the independent and correlated cases, respectively.
\end{corollary}

Next, we provide a trivial upper bound on the PoA derived from the SR gap.
This upper bound holds for all equilibrium concepts that generalize Bayes--Nash equilibria.

\begin{proposition}
\label{poa-ub-sr-gap}
In any Bayesian game, the price of anarchy for Bayes--Nash equilibria is at most the strategy representability gap.
\end{proposition}

\begin{proof}
Since every Bayes--Nash equilibrium can be expressed as a distribution over strategy profiles, its social welfare value is at most the value achieved by an optimal strategy profile.
By the definition of the SR gap, there exists a Bayesian game in which the social welfare achieved by an optimal strategy profile is at most the SR gap multiplied by the optimal social welfare.
Therefore, the PoA for Bayes--Nash equilibria in this game is at most the SR gap.
\end{proof}

From this proposition, we obtain the PoA upper bounds for Bayesian valid utility games.

\begin{corollary}
There exist $n$-player Bayesian valid utility games with independent and correlated prior distributions in which the price of anarchy for Bayes--Nash equilibria is at most $1/2$ and $\displaystyle O(1/ \sqrt{n} )$, respectively.
\end{corollary}

\subsection{For strategic-form coarse correlated equilibria}
\label{subsection-poa-sfcce}

We provide an improved PoA lower bound for SFCCEs, a subset of SFCBSs, under the assumption of independent priors.
The proof technique is a straightforward extension of the method developed by \citet*{Rou15} and \citet*{Syr12}.
However, to our knowledge, no one has provided an explicit proof, even for Bayes--Nash equilibria.
The proof is deferred to \Cref{proofs}.

\begin{restatable}{proposition}{poasfcce}
\label{poa-sfcce}
For any Bayesian valid utility game with an independent prior distribution, the price of anarchy for strategic-form coarse correlated equilibria is at least $1/2$.
\end{restatable}

This lower bound applies to other equilibrium concepts that are a subset of SFCCEs.
Since the PoA for (non-Bayesian) valid utility games can be $1/2$ even for mixed Nash equilibria \citep*{Vet}, this lower bound is tight for all equilibrium concepts between Bayes--Nash equilibria and SFCCEs.

\subsection{For communication equilibria}
\label{subsection-poa-comeq}
We provide a lower bound on the price of anarchy for communication equilibria in the independent case by extending the smoothness argument.
The proof of \Cref{poa-sfcce} and the existing smoothness arguments for Bayesian games \citep{Syr12,Roughgarden15,ST13,Fujii23} require that an equilibrium be a distribution over strategy profiles (strategy representability) to handle deviations independently of the type profile.
We bypass this barrier by using incentive constraints for untruthful type reporting.

\begin{proposition}
For any Bayesian valid utility game with an independent prior distribution, the price of anarchy for communication equilibria is at least $1/2$.
\end{proposition}

\begin{proof}
Let $\pi \in \prod_{\theta \in \Theta} \Delta(A^{\theta})$ be any communication equilibrium.
Fix any $\theta'_{-i} \in \Theta_{-i}$.
Let $a^\theta \in \argmax_{a \in A} \vSW(a)$ be an optimal action profile for each $\theta \in \Theta$.
For any $i \in N$ and types $\theta_i,\theta'_i \in \Theta_i$, by considering constant map $\phi(\cdot) = a_i^{\theta_i,\theta'_{-i}}$ in the incentive constraints for communication equilibria, we obtain
\begin{equation*}
\E_{\theta_{-i} \sim \rho|_{\theta_i}} \left[ \E_{a \sim \pi(\theta)} \left[ v_i(a) \right] \right]
\ge
\E_{\theta_{-i} \sim \rho|_{\theta_i}} \left[ \E_{a \sim \pi(\theta'_i,\theta_{-i})} \left[ v_i(a_i^{\theta_i,\theta'_{-i}}, a_{-i}) \right] \right].
\end{equation*}
Assume that $\theta' \sim \rho$ is a random variable independent of $\theta_{-i} \sim \rho|_{\theta_i}$.
By taking the expectation of this inequality for $\theta_i \sim \rho_i$ and $\theta' \sim \rho$, we obtain
\begin{equation*}
\E_{\theta \sim \rho} \left[ \E_{a \sim \pi(\theta)} \left[ v_i(a) \right] \right]
\ge
\E_{\theta' \sim \rho} \left[ \E_{\theta \sim \rho} \left[ \E_{a \sim \pi(\theta'_i,\theta_{-i})} \left[ v_i(a_i^{\theta_i,\theta'_{-i}}, a_{-i}) \right] \right] \right].
\end{equation*}
Since $\rho$ is a product distribution, we can swap random variables $\theta_{-i}$ and $\theta'_{-i}$ on the right-hand side.
Then we obtain
\begin{equation*}
\E_{\theta \sim \rho} \left[ \E_{a \sim \pi(\theta)} \left[ v_i(a) \right] \right]
\ge
\E_{\theta' \sim \rho} \left[ \E_{\theta \sim \rho} \left[ \E_{a \sim \pi(\theta')} \left[ v_i(a_i^{\theta}, a_{-i}) \right] \right] \right].
\end{equation*}
By taking the summation of this inequality for all $i \in N$, we obtain
\begin{align*}
&\E_{\theta \sim \rho} \left[ \E_{a \sim \pi(\theta)} \left[ \vSW(a) \right] \right]\\
&\ge
\sum_{i \in N} \E_{\theta \sim \rho} \left[ \E_{a \sim \pi(\theta)} \left[ v_i(a) \right] \right] \tag{from the total utility condition of valid utility games}\\
&\ge
\sum_{i \in N} \E_{\theta' \sim \rho} \left[ \E_{\theta \sim \rho} \left[ \E_{a \sim \pi(\theta')} \left[ v_i(a_i^{\theta}, a_{-i}) \right] \right] \right]\\
&\ge
\sum_{i \in N} \E_{\theta \sim \rho} \left[ \E_{\theta' \sim \rho} \left[ \E_{a \sim \pi(\theta')} \left[
f \left(a^\theta_i \;\middle|\; \{a_1,\dots,a_{i-1},a_{i+1},\dots,a_n\} \right)
\right] \right] \right] \tag{from the marginal contribution condition of valid utility games}\\
&\ge
\sum_{i \in N} \E_{\theta \sim \rho} \left[ \E_{\theta' \sim \rho} \left[ \E_{a \sim \pi(\theta')} \left[
f \left(a^\theta_i \;\middle|\; \{a_1,\dots,a_n\} \cup \left\{a_1^\theta, \dots, a_{i-1}^\theta\right\} \right)
\right] \right] \right] \tag{from submodularity}\\
&=
\E_{\theta \sim \rho} \left[ \E_{\theta' \sim \rho} \left[ \E_{a \sim \pi(\theta')} \left[
f \left( \left\{a_1^\theta,\dots,a_n^\theta \right\} \;\middle|\; \{a_1,\dots,a_n\}\right)
\right] \right] \right]\\
&\ge
\E_{\theta \sim \rho} \left[ \E_{\theta' \sim \rho} \left[ \E_{a \sim \pi(\theta')} \left[ 
f\left( \left\{a_1^\theta,\dots,a_n^\theta \right\} \right)
-
f( \{a_1,\dots,a_n\} )
\right] \right] \right] \tag{from monotonicity}\\
&=
\E_{\theta \sim \rho} \left[ 
\vSW\left( a^\theta \right)
\right]
-
\E_{\theta' \sim \rho} \left[ \E_{a \sim \pi(\theta')} \left[ 
\vSW(a)
\right] \right].
\end{align*}
Finally, recalling that $a^\theta \in \argmax_{a \in A} \vSW(a)$, we obtain
\begin{equation*}
\E_{\theta \sim \rho} \left[ \E_{a \sim \pi(\theta)} \left[ 
\vSW(a)
\right] \right]
\ge
\frac{1}{2}
\E_{\theta \sim \rho} \left[ 
\max_{a \in A} \vSW\left( a \right)
\right].
\qedhere
\end{equation*}
\end{proof}

\subsection{For Bayesian solutions}
\label{subsection-poa-bs}
The upper bound on the PoA for the independent case can be improved for Bayesian solutions as follows.
This result shows that the PoA of Bayesian solutions is strictly worse than that of SFCCEs and communication equilibria, for which better lower bounds of $1/2$ are proved in the previous two sections.

\begin{proposition}
There exists a Bayesian valid utility game with an independent prior distribution in which the price of anarchy for Bayesian solutions is at most $0.441$.
\end{proposition}

\begin{proof}
We consider a coverage function with universe $U$, where $|U| = n$.
There are two different groups of players.
The first group consists of $n/2$ players with low priority.
The action set of each player in this group is a singleton that contains only one element selected from $U$ uniformly at random.
These players have only one possible action, that is, they do not make any decision.
The second group consists of $n/2$ players with high priority.
Their action set always contains all the elements in $U$ (There is no randomness).
The payoff yielded by each element is distributed equally among high-priority players selecting it if at least one high-priority players select it, and distributed equally among low-priority players selecting it if only low-priority players select it.

In this Bayesian game, the optimal social welfare is $n \left( 1 - \left( 1 - \frac{1}{n} \right)^{n/2} \right) + \frac{n}{2}$ in expectation.
This can be achieved by making all the high-priority players select distinct elements that have not been selected by the low-priority players.
As the low-priority players have only one possible action, there is no option for their selection.
Each element is selected by at least one low-priority players with probability $1-(1-\frac{1}{n})^{n/2}$.
Since there are at least $n/2$ remaining elements, we can let the high-priority players disjointly select them.
In total, the expected number of the selected elements is
\begin{equation*}
n \left( 1 - \left( 1 - \frac{1}{n} \right)^{n/2} \right)
+
\frac{n}{2}.
\end{equation*}

Next, we show that there exists a Bayesian solution whose expected social welfare is $n \left( 1 - \left( 1 - \frac{1}{n} \right)^{n/2} \right)$.
A mediator chooses a perfect matching between the low-priority and high-priority players uniformly at random.
Then the mediator recommends the action selected by the matched low-priority player to each high-priority player.
Although the low-priority players always obtain payoff of $0$, they have no incentive to deviate from the recommendation because they have only one possible action.
If every player follows the mediator's recommendation, the high-priority players obtain the expected payoff $n \left( 1 - \left( 1 - \frac{1}{n} \right)^{n/2} \right) / (n/2) = 2 \left( 1 - \left( 1 - \frac{1}{n} \right)^{n/2} \right)$, because $n \left( 1 - \left( 1 - \frac{1}{n} \right)^{n/2} \right)$ elements are shared equally among the $n/2$ high-priority players in expectation.
Suppose that a high-priority player deviates from the mediator's recommendation.
Since the players' actions are determined by the uniform distribution, the only information each high-priority player has is that the recommended element is selected by some low-priority player.
Hence, the expected payoff of each element according to the posterior distribution for this player is the same as the recommended element, $2 \left( 1 - \left( 1 - \frac{1}{n} \right)^{n/2} \right)$.
Therefore, this player has no incentive to deviate, and the distribution realized by this mediator is a Bayesian solution.

The expected social welfare achieved by this Bayesian solution is $n \left( 1 - \left( 1 - \frac{1}{n} \right)^{n/2} \right)$.
The price of anarchy for Bayesian solutions is at most
\begin{equation*}
\frac{n \left( 1 - \left( 1 - \frac{1}{n} \right)^{n/2} \right)}{n \left( 1 - \left( 1 - \frac{1}{n} \right)^{n/2} \right) + \frac{n}{2}}
\to
\frac{1 - \exp \left( - \frac{1}{2} \right)}{\frac{3}{2} - \exp \left( - \frac{1}{2} \right)} \approx 0.4403
\end{equation*}
as $n \to \infty$.
\end{proof}

\begin{remark}
The constant $\frac{1 - \exp \left( - 1/2 \right)}{3/2 - \exp \left( - 1/2 \right)}$ can be slightly improved by selecting an appropriate proportion of high- and low-priority players.
However, that is still far from being tight compared to the lower bound $(1-1/\rme)/2$.
\end{remark}

Note that the distribution used in the above proof is not a communication equilibrium.
If a low-priority player untruthfully reports a random element to the mediator as if it were their possible action, the mediator recommends it to the matched high-priority player.
As a result, the dishonest low-priority player faces no interference from the matched high-priority player and may obtain a positive payoff.

\section{Price of stability in Bayesian basic utility games}
This section provides lower and upper bounds on the PoS in Bayesian basic utility games for various concepts of Bayes (coarse) correlated equilibria.
In \Cref{subsection-pos-bs}, we derive lower bounds on the PoS for Bayesian solutions.
In \Cref{subsection-pos-bne}, we derive lower and upper bounds on the PoS for Bayes--Nash equilibria using the strategic form.
In \Cref{subsection-pos-comeq}, we provide an example in which the PoS for communication equilibria is at most $4/5$.

\subsection{For Bayesian solutions}
\label{subsection-pos-bs}

First, we show that a PoS lower bound can be obtained from PoS lower bounds for complete-information games realized for each type profile.
This proposition holds for general Bayesian games, not restricted to Bayesian valid or basic utility games.

\begin{proposition}
\label{pos-general}
Let $\alpha$ be the price of stability for Bayesian solutions in a Bayesian game.
For each type profile $\theta \in \Theta$, let $\alpha^\theta$ be the price of stability for correlated equilibria in the complete-information game for fixed $\theta$.
Then $\alpha \ge \min_{\theta \in \Theta} \alpha^\theta$.
\end{proposition}

\begin{proof}
From the assumption, for any $\theta \in \Theta$, there exists a correlated equilibrium $\pi^\theta \in \Delta(A)$ for the complete-information game for $\theta$ that satisfies
\begin{equation}
\label{ce-approx}
\E_{a \sim \pi^\theta} \left[ \vSW(a) \right]
\ge
\alpha^\theta \max_{a \in A} \vSW(a).
\end{equation}

Define $\pi \in \prod_{\theta \in \Theta} \Delta(A^\theta)$ by $\pi(\theta) = \pi^\theta$ for each $\theta \in \Theta$.
We prove that $\pi$ is a Bayesian solution in the Bayesian game.
Fix any $i \in N$ and $\theta_i \in \Theta_i$.
Since $\pi^\theta$ is a correlated equilibrium for the complete-information game for each $\theta_{-i} \in \Theta_{-i}$, it holds
\begin{equation*}
\label{ce-ic}
\E_{a \sim \pi^\theta} \left[ v_i(a) \right]
\ge
\E_{a \sim \pi^\theta} \left[ v_i(\phi(a_i), a_{-i}) \right]
\end{equation*}
for any $i \in N$ and $\phi \colon A_i^{\theta_i} \to A_i^{\theta_i}$.
By taking the expectation in terms of $\theta_{-i} \sim \rho|_{\theta_i}$, we obtain
\begin{equation*}
\E_{\theta_{-i} \sim \rho|_{\theta_i}} \left[ \E_{a \sim \pi(\theta)} \left[ v_i(a) \right] \right]
\ge
\E_{\theta_{-i} \sim \rho|_{\theta_i}} \left[ \E_{a \sim \pi(\theta)} \left[ v_i(\phi(a_i), a_{-i}) \right] \right].
\end{equation*}
Hence, the incentive constraints for Bayesian solutions hold, and $\pi$ is a Bayesian solution.
By taking the expectation of both sides of \eqref{ce-approx}, we obtain
\begin{equation*}
\E_{\theta \sim \rho} \left[ \E_{a \sim \pi(\theta)} \left[ \vSW(a) \right] \right]
\ge
\E_{\theta \sim \rho} \left[ \alpha^\theta \max_{a \in A} \vSW(a) \right]
\ge
\min_{\theta \in \Theta} \alpha^\theta \E_{\theta \sim \rho} \left[ \max_{a \in A} \vSW(a) \right],
\end{equation*}
which implies the price of stability is at least $\min_{\theta \in \Theta}\alpha^\theta$.
\end{proof}

As \citet*{Vet} proved, (non-Bayesian) basic utility games are potential games with the social welfare function as a potential function.
Hence, a maximizer of the social welfare function is a pure Nash equilibrium.
This implies that there always exists a pure Nash equilibrium achieving the optimal social welfare in basic utility games, i.e., the PoS is $1$ for pure Nash equilibria, and correlated equilibria, which are a superset of pure Nash equilibria.
By combining this fact with the above theorem, we obtain the PoS lower bound for Bayesian basic utility games.
On the other hand, the PoS for correlated equilibria in (non-Bayesian) valid utility games is known to be $1/2$, which does not improve on the PoA.
Combining these existing results with \Cref{pos-general}, we obtain the following PoS bounds for Bayesian solutions.

\begin{corollary}
In the correlated case, the price of stability for Bayesian solutions is at least $1$ and $1/2$ in Bayesian basic and valid utility games, respectively.
\end{corollary}

Since there exists a (non-Bayesian) valid utility game whose PoS for coarse correlated equilibria is $1/2$ (see Example 8 of \citep*{MYK15}), this lower bound cannot be improved even for independent distributions.

\subsection{For Bayes--Nash equilibria}
\label{subsection-pos-bne}

Next, we present a PoS lower bound for Bayes--Nash equilibria in Bayesian basic utility games.
To prove it, we consider the strategic form, a complete-information interpretation of a Bayesian game.
In the strategic form, the set of strategies is interpreted as an action set for each player.
Formally, each player $i \in N$ selects a strategy $s_i \in S_i$ and then obtains payoff $\E_{\theta \sim \rho} \left[ v_i(s(\theta)) \right]$.
The value of social welfare for each strategy profile $s \in S$ is $\E_{\theta \sim \rho} \left[ \vSW(s(\theta)) \right]$.
We prove the following lemma, which enables us to use existing PoS lower bounds on (non-Bayesian) basic utility games.
The proof is deferred to \Cref{proofs}.

\begin{restatable}{lemma}{sfbasic}
\label{sf-basic}
The strategic form of any Bayesian basic utility game is a basic utility game.
\end{restatable}

By combining this lemma with the SR gap lower bounds, we obtain the PoS lower bound for Bayes--Nash equilibria.

\begin{proposition}
In Bayesian basic utility games, the price of stability for Bayes--Nash equilibria is at least the strategy representability gap.
\end{proposition}

\begin{proof}
From \Cref{sf-basic}, the strategic form of a Bayesian basic utility game is a basic utility game.
In basic utility games, the action profile that maximizes the social welfare function is a pure Nash equilibrium, i.e., 
\begin{equation*}
s^* \in \argmax_{s \in S} \E \left[ \vSW(s(\theta)) \right].
\end{equation*}
is a pure Nash equilibrium of the strategic form.
It is known that the set of Bayes--Nash equilibria of a Bayesian game is equivalent to the set of Nash equilibria of its strategic form \citep*{Harsanyi67,Harsanyi68a,Harsanyi68b}.
Hence, $s^*$ is a Bayes--Nash equilibrium of the original Bayesian basic utility game.
Recall that the SR gap represents the ratio of the optimal social welfare to the value achieved by an optimal strategy profile.
Therefore, there exists a Bayes--Nash equilibrium that achieves a social welfare value of the SR gap multiplied by the optimal social welfare.
\end{proof}

The following proposition shows that this lower bound for Bayes--Nash equilibria is tight, even for SFCCEs.
This follows from the general relation between the SR gap and the PoS for SFCCEs in Bayesian games.

\begin{proposition}
In any Bayesian game, the price of stability for strategic-form coarse correlated equilibria is at most the strategy representability gap.
\end{proposition}

\begin{proof}
Since any SFCCE is a distribution over strategy profiles, the social welfare achieved by it is upper-bounded by the value achieved by an optimal strategy profile.
Therefore, their ratio to the optimal social welfare is no larger than the SR gap.
\end{proof}

As a corollary of these results, we obtain PoS lower and upper bounds for Bayes--Nash equilibria.

\begin{corollary}
In basic utility games, the price of stability for Bayes--Nash equilibria is at least $1-1/\rme$ in the independent case and $\Omega( 1 / \sqrt{n} )$ in the correlated case, respectively.
These bounds are tight (up to a multiplicative constant in the correlated case) even for strategic-form coarse correlated equilibria.
\end{corollary}

\begin{figure}
\centering
\begin{tikzpicture}[line width=1pt]
\node (u1) [draw, circle, minimum width=0.05\paperwidth] at (0, 0) {$2$};
\node (u2) [draw, circle, minimum width=0.03\paperwidth] at (2, 0) {$1$};
\node (u3) [draw, circle, minimum width=0.01\paperwidth, inner sep=2pt] at (3.5, 0) {\scriptsize $\epsilon$};
\node at (u1.south) [anchor=north] {$u_1$};
\node at (u2.south) [anchor=north] {$u_2$};
\node at (u3.south) [anchor=north] {$u_3$};
\node (a2) at (1, 1) {$a_2$};
\node (a2p) at (2.5, 1) {$a'_2$};
\node (t1) at (3.5, 1) {$\theta_1$};
\node (t1p) at (0, 1) {$\theta'_1$};
\draw (a2) -- (u1);
\draw (a2) -- (u3);
\draw (a2p) -- (u2);
\draw (t1) -- (u3);
\draw (t1p) -- (u1);
\end{tikzpicture}
\caption{An example of a Bayesian basic utility game in which an optimal communication equilibrium is suboptimal.}
\label{figure-com}
\end{figure}
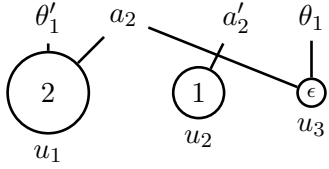

\subsection{For communication equilibria}
\label{subsection-pos-comeq}

Recall that the PoS for Bayesian solutions is always $1$ in Bayesian basic utility games.
In contrast, the PoS for equilibrium concepts between Bayes--Nash equilibria and SFCCEs coincides with the SR gap, hence suboptimal.
The PoS for communication equilibria is still unclear.
The following example shows that the PoS for communication equilibria can be smaller than $1$ even in the independent case.

\begin{proposition}
There exists a Bayesian basic utility game with an independent prior distribution in which the price of stability for communication equilibria is at most $4/5$.
\end{proposition}

\begin{proof}
We provide an example of a Bayesian basic utility game with two players (\Cref{figure-com}).
The first player has two types and one action, and the second player has one type and two actions.
Let $\Theta_1 = \{\theta_1,\theta'_1\}$ and $A_2 = \{a_2,a'_2\}$.
Assume the prior distribution is uniform over $\Theta_1$.
We consider a weighted coverage function with universe $U = \{u_1,u_2,u_3\}$ and weights $w(u_1) = 2$, $w(u_2) = 1$, and $w(u_3) = \epsilon$.
The first player with type $\theta_1$ covers $u_3$ and $\theta'_1$ covers $u_1$ (Note that the first player has only one action and does not have any option).
The second player's action $a_2$ covers $u_1$ and $u_3$, and $a'_2$ covers $u_2$.
The optimal social welfare is $(5+\epsilon)/2$, which is achieved when the second player chooses $a_2$ if the first player's type is $\theta_1$ and $a'_2$ if $\theta'_1$.

We prove that any communication equilibrium achieves at most $4/5$ of the social welfare.
Let $p \in [0,1]$ be the probability that the mediator recommends $a_2$ when $\theta_1$ is reported by the first player, and $q \in [0,1]$ the probability that the mediator recommends $a_2$ when $\theta'_1$ is reported.
From the incentive constraints of communication equilibria, the first player must not be better off by untruthfully reporting the type.
The first player with type $\theta_1$ obtains a payoff $\epsilon$ if the second player chooses $a'_2$ and $0$ if the second player chooses $a_2$.
Hence, if $p > q$, the first player reports an untrue type, and hence $p \le q$ must hold.
On the other hand, the first player with type $\theta'_1$ obtains payoff $2$ if the second player chooses $a'_2$ and $0$ if the second player chooses $a_2$.
Hence, if $p < q$, the first player becomes untruthful, and hence $p \ge q$ must hold.
Therefore, any communication equilibrium must satisfy $p=q$.

Furthermore, the incentive constraint for the second player requires the recommendation to the second player to be incentive-compatible.
Since $p=q$, when $a'_2$ is recommended, the posterior distribution of the first player's type is uniform.
However, the deviation to $a_2$ yields $(2+\epsilon)/2$, while following the recommendation obtains only $1$.
Therefore, $p=q=1$ is the only communication equilibrium in this game.
The social welfare achieved by it is $2+\epsilon$, which is $(4/5+o(1))$-approximation to the social optimal.
\end{proof}

\section*{Acknowledgment}
Kaito Fujii was supported by JSPS KAKENHI Grant Number 22K17857.

\bibliographystyle{ACM-Reference-Format}
\bibliography{main}


\begin{thebibliography}{32}


\ifx \showCODEN    \undefined \def \showCODEN     #1{\unskip}     \fi
\ifx \showDOI      \undefined \def \showDOI       #1{#1}\fi
\ifx \showISBNx    \undefined \def \showISBNx     #1{\unskip}     \fi
\ifx \showISBNxiii \undefined \def \showISBNxiii  #1{\unskip}     \fi
\ifx \showISSN     \undefined \def \showISSN      #1{\unskip}     \fi
\ifx \showLCCN     \undefined \def \showLCCN      #1{\unskip}     \fi
\ifx \shownote     \undefined \def \shownote      #1{#1}          \fi
\ifx \showarticletitle \undefined \def \showarticletitle #1{#1}   \fi
\ifx \showURL      \undefined \def \showURL       {\relax}        \fi
\providecommand\bibfield[2]{#2}
\providecommand\bibinfo[2]{#2}
\providecommand\natexlab[1]{#1}
\providecommand\showeprint[2][]{arXiv:#2}

\bibitem[Agrawal et~al\mbox{.}(2012)]%
        {ADSY12}
\bibfield{author}{\bibinfo{person}{Shipra Agrawal}, \bibinfo{person}{Yichuan Ding}, \bibinfo{person}{Amin Saberi}, {and} \bibinfo{person}{Yinyu Ye}.} \bibinfo{year}{2012}\natexlab{}.
\newblock \showarticletitle{Price of Correlations in Stochastic Optimization}.
\newblock \bibinfo{journal}{\emph{Operations Research}} \bibinfo{volume}{60}, \bibinfo{number}{1} (\bibinfo{year}{2012}), \bibinfo{pages}{150--162}.
\newblock


\bibitem[Anshelevich et~al\mbox{.}(2008)]%
        {ADKTWR08}
\bibfield{author}{\bibinfo{person}{Elliot Anshelevich}, \bibinfo{person}{Anirban Dasgupta}, \bibinfo{person}{Jon~M. Kleinberg}, \bibinfo{person}{{\'{E}}va Tardos}, \bibinfo{person}{Tom Wexler}, {and} \bibinfo{person}{Tim Roughgarden}.} \bibinfo{year}{2008}\natexlab{}.
\newblock \showarticletitle{The Price of Stability for Network Design with Fair Cost Allocation}.
\newblock \bibinfo{journal}{\emph{SIAM J. Comput.}} \bibinfo{volume}{38}, \bibinfo{number}{4} (\bibinfo{year}{2008}), \bibinfo{pages}{1602--1623}.
\newblock


\bibitem[Arieli and Babichenko(2019)]%
        {AB19}
\bibfield{author}{\bibinfo{person}{Itai Arieli} {and} \bibinfo{person}{Yakov Babichenko}.} \bibinfo{year}{2019}\natexlab{}.
\newblock \showarticletitle{Private {B}ayesian persuasion}.
\newblock \bibinfo{journal}{\emph{Journal of Economic Theory}}  \bibinfo{volume}{182} (\bibinfo{year}{2019}), \bibinfo{pages}{185--217}.
\newblock


\bibitem[Babichenko and Barman(2017)]%
        {BB17}
\bibfield{author}{\bibinfo{person}{Yakov Babichenko} {and} \bibinfo{person}{Siddharth Barman}.} \bibinfo{year}{2017}\natexlab{}.
\newblock \showarticletitle{{Algorithmic aspects of private {B}ayesian persuasion}}. In \bibinfo{booktitle}{\emph{8th Innovations in Theoretical Computer Science Conference (ITCS 2017)}}. \bibinfo{pages}{34:1--34:16}.
\newblock


\bibitem[Ben{-}Porat and Tennenholtz(2017)]%
        {BT17}
\bibfield{author}{\bibinfo{person}{Omer Ben{-}Porat} {and} \bibinfo{person}{Moshe Tennenholtz}.} \bibinfo{year}{2017}\natexlab{}.
\newblock \showarticletitle{Shapley Facility Location Games}. In \bibinfo{booktitle}{\emph{Proceedings of The 13th International Conference on Web and Internet Economics ({WINE} 2017)}}. \bibinfo{pages}{58--73}.
\newblock


\bibitem[Bergemann and Morris(2019)]%
        {BM19}
\bibfield{author}{\bibinfo{person}{Dirk Bergemann} {and} \bibinfo{person}{Stephen Morris}.} \bibinfo{year}{2019}\natexlab{}.
\newblock \showarticletitle{Information Design: A Unified Perspective}.
\newblock \bibinfo{journal}{\emph{Journal of Economic Literature}} \bibinfo{volume}{57}, \bibinfo{number}{1} (\bibinfo{year}{2019}), \bibinfo{pages}{44--95}.
\newblock


\bibitem[Bharathi et~al\mbox{.}(2007)]%
        {BKS07}
\bibfield{author}{\bibinfo{person}{Shishir Bharathi}, \bibinfo{person}{David Kempe}, {and} \bibinfo{person}{Mahyar Salek}.} \bibinfo{year}{2007}\natexlab{}.
\newblock \showarticletitle{Competitive Influence Maximization in Social Networks}. In \bibinfo{booktitle}{\emph{Proceedings of the Third International Workshop on Internet and Network Economics ({WINE} 2007)}}. \bibinfo{pages}{306--311}.
\newblock


\bibitem[C{\u{a}}linescu et~al\mbox{.}(2011)]%
        {CCPV11}
\bibfield{author}{\bibinfo{person}{Gruia C{\u{a}}linescu}, \bibinfo{person}{Chandra Chekuri}, \bibinfo{person}{Martin P{\'{a}}l}, {and} \bibinfo{person}{Jan Vondr{\'{a}}k}.} \bibinfo{year}{2011}\natexlab{}.
\newblock \showarticletitle{Maximizing a Monotone Submodular Function Subject to a Matroid Constraint}.
\newblock \bibinfo{journal}{\emph{SIAM J. Comput.}} \bibinfo{volume}{40}, \bibinfo{number}{6} (\bibinfo{year}{2011}), \bibinfo{pages}{1740--1766}.
\newblock


\bibitem[Caragiannis et~al\mbox{.}(2015)]%
        {CKK15}
\bibfield{author}{\bibinfo{person}{Ioannis Caragiannis}, \bibinfo{person}{Christos Kaklamanis}, \bibinfo{person}{Panagiotis Kanellopoulos}, \bibinfo{person}{Maria Kyropoulou}, \bibinfo{person}{Brendan Lucier}, \bibinfo{person}{Renato~Paes Leme}, {and} \bibinfo{person}{{\'{E}}va Tardos}.} \bibinfo{year}{2015}\natexlab{}.
\newblock \showarticletitle{Bounding the inefficiency of outcomes in generalized second price auctions}.
\newblock \bibinfo{journal}{\emph{Journal of Economic Theory}}  \bibinfo{volume}{156} (\bibinfo{year}{2015}), \bibinfo{pages}{343--388}.
\newblock


\bibitem[Dughmi and Xu(2017)]%
        {DX17}
\bibfield{author}{\bibinfo{person}{Shaddin Dughmi} {and} \bibinfo{person}{Haifeng Xu}.} \bibinfo{year}{2017}\natexlab{}.
\newblock \showarticletitle{Algorithmic Persuasion with No Externalities}. In \bibinfo{booktitle}{\emph{Proceedings of the 2017 {ACM} Conference on Economics and Computation ({EC} 2017)}}. \bibinfo{publisher}{{ACM}}, \bibinfo{pages}{351--368}.
\newblock


\bibitem[Forges(1993)]%
        {Forges93}
\bibfield{author}{\bibinfo{person}{Fran\c{c}oise Forges}.} \bibinfo{year}{1993}\natexlab{}.
\newblock \showarticletitle{Five legitimate definitions of correlated equilibrium in games with incomplete information}.
\newblock \bibinfo{journal}{\emph{Theory and Decision}}  \bibinfo{volume}{35} (\bibinfo{year}{1993}), \bibinfo{pages}{277--310}.
\newblock


\bibitem[Forges(2006)]%
        {Forges06}
\bibfield{author}{\bibinfo{person}{Fran\c{c}oise Forges}.} \bibinfo{year}{2006}\natexlab{}.
\newblock \showarticletitle{Correlated equilibrium in games with incomplete information revisited}.
\newblock \bibinfo{journal}{\emph{Theory and Decision}}  \bibinfo{volume}{61} (\bibinfo{year}{2006}), \bibinfo{pages}{329--344}.
\newblock


\bibitem[Frieze and Karo\'{n}ski(2015)]%
        {FK15}
\bibfield{author}{\bibinfo{person}{Alan Frieze} {and} \bibinfo{person}{Micha\l Karo\'{n}ski}.} \bibinfo{year}{2015}\natexlab{}.
\newblock \bibinfo{booktitle}{\emph{Introduction to Random Graphs}}.
\newblock \bibinfo{publisher}{Cambridge University Press}.
\newblock


\bibitem[Fujii(2023)]%
        {Fujii23}
\bibfield{author}{\bibinfo{person}{Kaito Fujii}.} \bibinfo{year}{2023}\natexlab{}.
\newblock \showarticletitle{Bayes correlated equilibria and no-regret dynamics}.
\newblock  (\bibinfo{year}{2023}).
\newblock
\showeprint[arxiv]{2304.05005}


\bibitem[Fujii and Sakaue(2022)]%
        {FS22}
\bibfield{author}{\bibinfo{person}{Kaito Fujii} {and} \bibinfo{person}{Shinsaku Sakaue}.} \bibinfo{year}{2022}\natexlab{}.
\newblock \showarticletitle{Algorithmic {B}ayesian persuasion with combinatorial actions}. In \bibinfo{booktitle}{\emph{Thirty-Sixth {AAAI} Conference on Artificial Intelligence ({AAAI} 2022)}}. \bibinfo{pages}{5016--5024}.
\newblock


\bibitem[Harsanyi(1967)]%
        {Harsanyi67}
\bibfield{author}{\bibinfo{person}{John~C. Harsanyi}.} \bibinfo{year}{1967}\natexlab{}.
\newblock \showarticletitle{Games with Incomplete Information Played by ``{B}ayesian'' Players, Part {I}. {T}he Basic Model}.
\newblock \bibinfo{journal}{\emph{Management Science}} \bibinfo{volume}{14}, \bibinfo{number}{3} (\bibinfo{year}{1967}), \bibinfo{pages}{159--182}.
\newblock


\bibitem[Harsanyi(1968a)]%
        {Harsanyi68a}
\bibfield{author}{\bibinfo{person}{John~C. Harsanyi}.} \bibinfo{year}{1968}\natexlab{a}.
\newblock \showarticletitle{Games with Incomplete Information Played by ``{B}ayesian'' Players Part {II}. {B}ayesian Equilibrium Points}.
\newblock \bibinfo{journal}{\emph{Management Science}} \bibinfo{volume}{14}, \bibinfo{number}{5} (\bibinfo{year}{1968}), \bibinfo{pages}{320--334}.
\newblock


\bibitem[Harsanyi(1968b)]%
        {Harsanyi68b}
\bibfield{author}{\bibinfo{person}{John~C. Harsanyi}.} \bibinfo{year}{1968}\natexlab{b}.
\newblock \showarticletitle{Games with Incomplete Information Played by `{B}ayesian' Players, Part {III}. {T}he Basic Probability Distribution of the Game}.
\newblock \bibinfo{journal}{\emph{Management Science}} \bibinfo{volume}{14}, \bibinfo{number}{7} (\bibinfo{year}{1968}), \bibinfo{pages}{486--502}.
\newblock


\bibitem[Hartline et~al\mbox{.}(2015)]%
        {HST15}
\bibfield{author}{\bibinfo{person}{Jason~D. Hartline}, \bibinfo{person}{Vasilis Syrgkanis}, {and} \bibinfo{person}{{\'{E}}va Tardos}.} \bibinfo{year}{2015}\natexlab{}.
\newblock \showarticletitle{No-regret learning in {B}ayesian games}. In \bibinfo{booktitle}{\emph{Advances in Neural Information Processing Systems 28 ({NIPS} 2015)}}. \bibinfo{pages}{3061--3069}.
\newblock


\bibitem[He and Kempe(2013)]%
        {HK13}
\bibfield{author}{\bibinfo{person}{Xinran He} {and} \bibinfo{person}{David Kempe}.} \bibinfo{year}{2013}\natexlab{}.
\newblock \showarticletitle{Price of Anarchy for the $N$-Player competitive Cascade Game with Submodular Activation Functions}. In \bibinfo{booktitle}{\emph{Proceedings of the Ninth International Conference on Web and Internet Economics ({WINE 2013})}}. \bibinfo{pages}{232--248}.
\newblock


\bibitem[Jin and Lu(2023)]%
        {JL23}
\bibfield{author}{\bibinfo{person}{Yaonan Jin} {and} \bibinfo{person}{Pinyan Lu}.} \bibinfo{year}{2023}\natexlab{}.
\newblock \showarticletitle{The price of stability for first price auction}. In \bibinfo{booktitle}{\emph{Proceedings of the 2023 {ACM-SIAM} Symposium on Discrete Algorithms ({SODA} 2023)}}. \bibinfo{pages}{332--352}.
\newblock


\bibitem[Koutsoupias and Papadimitriou(1999)]%
        {KP99}
\bibfield{author}{\bibinfo{person}{Elias Koutsoupias} {and} \bibinfo{person}{Christos~H. Papadimitriou}.} \bibinfo{year}{1999}\natexlab{}.
\newblock \showarticletitle{Worst-case equilibria}. In \bibinfo{booktitle}{\emph{Proceedings of the 16th Annual Symposium on Theoretical Aspects of Computer Science ({STACS} 1999)}}. \bibinfo{pages}{404--413}.
\newblock


\bibitem[Maehara et~al\mbox{.}(2015)]%
        {MYK15}
\bibfield{author}{\bibinfo{person}{Takanori Maehara}, \bibinfo{person}{Akihiro Yabe}, {and} \bibinfo{person}{Ken{-}ichi Kawarabayashi}.} \bibinfo{year}{2015}\natexlab{}.
\newblock \showarticletitle{Budget Allocation Problem with Multiple Advertisers: {A} Game Theoretic View}. In \bibinfo{booktitle}{\emph{Proceedings of the 32nd International Conference on Machine Learning (ICML 2015)}}. \bibinfo{pages}{428--437}.
\newblock


\bibitem[Myerson(1982)]%
        {Myerson82}
\bibfield{author}{\bibinfo{person}{Roger~B Myerson}.} \bibinfo{year}{1982}\natexlab{}.
\newblock \showarticletitle{Optimal coordination mechanisms in generalized principal–agent problems}.
\newblock \bibinfo{journal}{\emph{Journal of Mathematical Economics}} \bibinfo{volume}{10}, \bibinfo{number}{1} (\bibinfo{year}{1982}), \bibinfo{pages}{67--81}.
\newblock


\bibitem[Myerson(1997)]%
        {Myerson}
\bibfield{author}{\bibinfo{person}{Roger~B. Myerson}.} \bibinfo{year}{1997}\natexlab{}.
\newblock \bibinfo{booktitle}{\emph{Game Theory: Analysis of Conflict}}.
\newblock \bibinfo{publisher}{Harvard University Press}.
\newblock


\bibitem[Qiu and Singla(2022)]%
        {QS22}
\bibfield{author}{\bibinfo{person}{Frederick Qiu} {and} \bibinfo{person}{Sahil Singla}.} \bibinfo{year}{2022}\natexlab{}.
\newblock \showarticletitle{{Submodular Dominance and Applications}}. In \bibinfo{booktitle}{\emph{Approximation, Randomization, and Combinatorial Optimization. Algorithms and Techniques (APPROX/RANDOM 2022)}}, Vol.~\bibinfo{volume}{245}. \bibinfo{pages}{44:1--44:21}.
\newblock


\bibitem[Roughgarden(2015a)]%
        {Roughgarden15}
\bibfield{author}{\bibinfo{person}{Tim Roughgarden}.} \bibinfo{year}{2015}\natexlab{a}.
\newblock \showarticletitle{Intrinsic robustness of the price of anarchy}.
\newblock \bibinfo{journal}{\emph{Journal of the {ACM}}} \bibinfo{volume}{62}, \bibinfo{number}{5} (\bibinfo{year}{2015}), \bibinfo{pages}{32:1--32:42}.
\newblock


\bibitem[Roughgarden(2015b)]%
        {Rou15}
\bibfield{author}{\bibinfo{person}{Tim Roughgarden}.} \bibinfo{year}{2015}\natexlab{b}.
\newblock \showarticletitle{The price of anarchy in games of incomplete information}.
\newblock \bibinfo{journal}{\emph{{ACM} Transactions on Economics and Computation}} \bibinfo{volume}{3}, \bibinfo{number}{1} (\bibinfo{year}{2015}), \bibinfo{pages}{6:1--6:20}.
\newblock


\bibitem[Syrgkanis(2012)]%
        {Syr12}
\bibfield{author}{\bibinfo{person}{Vasilis Syrgkanis}.} \bibinfo{year}{2012}\natexlab{}.
\newblock \showarticletitle{{B}ayesian games and the smoothness framework}.
\newblock \bibinfo{journal}{\emph{Ar{X}iv preprint}}  \bibinfo{volume}{ar{X}iv:1203.5155} (\bibinfo{year}{2012}).
\newblock


\bibitem[Syrgkanis and Tardos(2013)]%
        {ST13}
\bibfield{author}{\bibinfo{person}{Vasilis Syrgkanis} {and} \bibinfo{person}{{\'{E}}va Tardos}.} \bibinfo{year}{2013}\natexlab{}.
\newblock \showarticletitle{Composable and efficient mechanisms}. In \bibinfo{booktitle}{\emph{Proceedings of the 45th Annual ACM Symposium on Theory of Computing ({STOC} 2013)}}. \bibinfo{pages}{211--220}.
\newblock


\bibitem[Vetta(2002)]%
        {Vet}
\bibfield{author}{\bibinfo{person}{Adrian Vetta}.} \bibinfo{year}{2002}\natexlab{}.
\newblock \showarticletitle{Nash equilibria in competitive societies, with applications to facility location, traffic routing and auctions}. In \bibinfo{booktitle}{\emph{Proceedings of The 43rd Annual IEEE Symposium on Foundations of Computer Science (FOCS 2002)}}. \bibinfo{pages}{416--425}.
\newblock


\bibitem[Vondr{\'a}k(2007)]%
        {Von07}
\bibfield{author}{\bibinfo{person}{Jan Vondr{\'a}k}.} \bibinfo{year}{2007}\natexlab{}.
\newblock \emph{\bibinfo{title}{Submodularity in combinatorial optimization}}.
\newblock \bibinfo{thesistype}{Ph.\,D. Dissertation}. \bibinfo{school}{Charles University, Prague}.
\newblock


\end{thebibliography}

\appendix

\section{Omitted proofs}\label{proofs}

\poacbs*

\begin{proof}
Let $s^* \in \argmax_{s \in S} \E \left[ \vSW(s(\theta)) \right]$.
Let $\pi \in \prod_{\theta \in \Theta} \Delta(A^\theta)$ be any SFCBS.
From its definition, we have
\begin{align*}
\E_{\theta \sim \rho} \left[ \E_{a \sim \pi(\theta)} \left[ v_i(a) \right] \right] 
\ge
\E_{\theta \sim \rho} \left[ \E_{a \sim \pi(\theta)} \left[ v_i(s_i^*(\theta_i), a_{-i}) \right] \right] 
\end{align*}
for any $i \in N$.
Then we obtain
\begin{align*}
&\E_{\theta \sim \rho} \left[ \E_{a \sim \pi(\theta)} \left[ \vSW(a) \right] \right]\\
&\ge
\sum_{i \in N} \E_{\theta \sim \rho} \left[ \E_{a \sim \pi(\theta)} \left[ v_i(a) \right] \right] \tag{from the total utility condition of valid utility games}\\
&\ge
\sum_{i \in N} \E_{\theta \sim \rho} \left[ \E_{a \sim \pi(\theta)} \left[ v_i(s^*_i(\theta_i),a_{-i}) \right] \right] \tag{from the property of SFCBSs}\\
&\ge
\sum_{i \in N} \E_{\theta \sim \rho} \left[ \E_{a \sim \pi(\theta)} \left[ \vSW(s^*_i(\theta_i),a_{-i}) - \vSW(\emptyset_i,a_{-i}) \right] \right] \tag{from the marginal contribution condition of valid utility games}\\
&=
\sum_{i \in N} \E_{\theta \sim \rho} \left[ \E_{a \sim \pi(\theta)} \left[ f \left( s^*_i(\theta_i) \;\middle|\; \{a_1,\dots,a_{i-1},a_{i+1},\dots,a_n\} \right) \right] \right] \tag{from the definition of $\vSW$}\\
&\ge
\sum_{i=1}^n \E_{\theta \sim \rho} \left[ \E_{a \sim \pi(\theta)} \left[ f \left( s^*_i(\theta_i) \;\middle|\; \{a_1,\dots,a_n\} \cup \{ s^*_1(\theta_1),\dots,s^*_{i-1}(\theta_{i-1}) \} \right) \right] \right] \tag{from submodularity}\\
&=
\E_{\theta \sim \rho} \left[ \E_{a \sim \pi(\theta)} \left[ f \left( \left\{s_1^*(\theta_1),\dots,s_n^*(\theta_n)\right\} \;\middle|\; \{a_1,\dots,a_n\} \right) \right] \right]\\
&\ge
\E_{\theta \sim \rho} \left[ \E_{a \sim \pi(\theta)} \left[ f \left( \left\{s_1^*(\theta_1),\dots,s_n^*(\theta_n)\right\} \right) - f \left( \{a_1,\dots,a_n\} \right) \right] \right] \tag{from monotonicity}\\
&=
\E_{\theta \sim \rho} \left[ \vSW(s^*(\theta)) \right] - \E_{\theta \sim \rho} \left[ \E_{a \sim \pi(\theta)} \left[ \vSW(a) \right] \right]. \tag{from the definition of $\vSW$}
\end{align*}
Rearranging this inequality, we obtain $\E_{\theta \sim \rho} \left[ \E_{a \sim \pi(\theta)} \left[ \vSW(a) \right] \right] \ge \frac{1}{2} \E_{\theta \sim \rho} \left[ \vSW(s^*(\theta)) \right]$.
\end{proof}

\poasfcce*

\begin{proof}
Let $\sigma \in \Delta(S)$ be any SFCCE.
From the definition of SFCCEs, for any $i \in N$ and strategy $s^*_i \in S_i$, it holds that
\begin{equation*}
\E_{\theta \sim \rho} \left[ \E_{s \sim \sigma} \left[ v_i(s(\theta)) \right] \right]
\ge
\E_{\theta \sim \rho} \left[ \E_{s \sim \sigma} \left[ v_i(s^*_i(\theta_i), s_{-i}(\theta_{-i})) \right] \right].
\end{equation*}
Let $a^\theta \in \argmax_{a \in A} \vSW(a)$ be an optimal action profile for each $\theta \in \Theta$, and we set $s^*_i(\theta_i) = a^{\theta_i, \theta'_{-i}}$.
We introduce another random variable $\theta' \in \Theta$ that is also generated from $\rho$ independent of $\theta$ and take the expected value of both-hand sides.
Note that $\theta'_{-i}$ is a random variable not depending on $\theta_{-i}$.
Since $\rho$ is a product distribution, we can swap random variables $\theta_{-i}$ and $\theta'_{-i}$ on the right-hand side.
Then we obtain
\begin{equation}\label{variable-swap}
\E_{\theta \sim \rho} \left[ \E_{s \sim \sigma} \left[ v_i(s(\theta)) \right] \right]
\ge
\E_{\theta \sim \rho} \left[ \E_{\theta' \sim \rho} \left[ \E_{s \sim \sigma} \left[ v_i(a^\theta_i, s_{-i}(\theta'_{-i})) \right] \right] \right].
\end{equation}
Using this inequality, we obtain
\begin{align*}
&\E_{\theta \sim \rho} \left[ \E_{s \sim \sigma} \left[ \vSW(s(\theta)) \right] \right]\\
&\ge
\sum_{i \in N} \E_{\theta \sim \rho} \left[ \E_{s \sim \sigma} \left[ v_i(s(\theta)) \right] \right] \tag{from the total utility condition of valid utility games}\\
&\ge
\sum_{i \in N} \E_{\theta \sim \rho} \left[ \E_{\theta' \sim \rho} \left[ \E_{s \sim \sigma} \left[ v_i(a^\theta_i, s_{-i}(\theta'_{-i}))
\right] \right] \right] \tag{from \eqref{variable-swap}}\\
&\ge
\sum_{i \in N} \E_{\theta \sim \rho} \left[ \E_{\theta' \sim \rho} \left[ \E_{s \sim \sigma} \left[
f \left(a^\theta_i \;\middle|\; \left\{s_1(\theta'_1),\dots,s_{i-1}(\theta'_{i-1}),s_{i+1}(\theta'_{i+1}),\dots,s_n(\theta'_n)\right\} \right)
\right] \right] \right] \tag{from the marginal contribution condition of valid utility games}\\
&\ge
\sum_{i=1}^n \E_{\theta \sim \rho} \left[ \E_{\theta' \sim \rho} \left[ \E_{s \sim \sigma} \left[
f \left(a^\theta_i \;\middle|\; \left\{ s_1(\theta'_1),\dots,s_n(\theta'_n) \right\} \cup \{a_1^\theta,\dots,a_{i-1}^\theta\} \right)
\right] \right] \right] \tag{from submodularity}\\
&=
\E_{\theta \sim \rho} \left[ \E_{\theta' \sim \rho} \left[ \E_{s \sim \sigma} \left[
f \left(\{a_1^\theta,\dots,a_n^\theta\} \;\middle|\; \{ s_1(\theta'_1),\dots,s_n(\theta'_n) \}\right)
\right] \right] \right]\\
&\ge
\E_{\theta \sim \rho} \left[ \E_{\theta' \sim \rho} \left[ \E_{s \sim \sigma} \left[ 
f\left( \left\{ a_1^\theta,\dots,a_n^\theta \right\} \right)
-
f\left( \left\{ s_1(\theta'_1),\dots,s_n(\theta'_n) \right\}\right)
\right] \right] \right] \tag{from monotonicity}\\
&=
\E_{\theta \sim \rho} \left[ 
\vSW\left( a^\theta \right)
\right]
-
\E_{\theta' \sim \rho} \left[ \E_{s \sim \sigma} \left[ 
\vSW(s(\theta'))
\right] \right].
\end{align*}
Finally, recalling that $a^\theta \in \argmax_{a \in A} \vSW(a)$, we obtain
\begin{equation*}
\E_{\theta \sim \rho} \left[ \E_{s \sim \sigma} \left[ 
\vSW(s(\theta))
\right] \right]
\ge
\frac{1}{2}
\E_{\theta \sim \rho} \left[ 
\max_{a \in A} \vSW\left( a \right)
\right].
\qedhere
\end{equation*}
\end{proof}

\sfbasic*

\begin{proof}
We verify that the strategic form satisfies the properties of basic utility games.
Let $E' = \bigcup_{i \in N} S_i$ be the set of all players' strategies.
We define $f' \colon 2^{E'} \to \bbR_{\ge 0}$ as the set function representing the social welfare by
\begin{equation*}
f'(X) = \E_{\theta \sim \rho} \left[ f\left( \bigcup_{i \in N} \{ s_i(\theta_i) \mid s_i \in X \cap S_i \} \right) \right]
\end{equation*}
for each $X \subseteq E'$,
where $f$ is the non-negative monotone submodular function for the original Bayesian basic utility game.
For each $s \in S$, we have $f'(\{s_1,\dots,s_n\}) = \E_{\theta \sim \rho} \left[ \vSW(s(\theta)) \right]$.
For each fixed $\theta \in \Theta$, the set function inside the expectation interprets each strategy $s_i \in S_i$ as an element $s_i(\theta_i) \in E$, hence being non-negative, monotone, and submodular.
Moreover, since taking the expectation does not violate these properties, $f'$ is also non-negative, monotone, and submodular.

The total utility condition $
\sum_{i \in N} \E_{\theta \sim \rho} \left[ v_i(s(\theta)) \right]
\le
\E_{\theta \sim \rho} \left[ \vSW(s(\theta)) \right]
$
for any $s \in S$
can be obtained by taking the expectation of
the total utility condition for the original Bayesian basic utility game.
The marginal contribution condition
$
\E_{\theta \sim \rho} \left[ \vSW(s(\theta)) \right] - 
\E_{\theta \sim \rho} \left[ \vSW(\emptyset_i, s_{-i}(\theta_{-i})) \right] 
=
\E_{\theta \sim \rho} \left[ v_i(s(\theta)) \right] 
$
for any $s \in S$
also directly follows by taking the expectation of the marginal contribution condition of the original Bayesian basic utility game.
\end{proof}

\section{On Bayesian valid/basic utility games}

\subsection{On the formulation}\label{setting}

In the standard formulation of Bayesian games, the set of actions $\tilde{A}_i$ is fixed for each player $i \in N$, and the utility function $\tilde{v}_i \colon \Theta \times \tilde{A} \to \bbR_{\ge 0}$ depends not only on action profile $\tilde{a} \in \tilde{A}$ but also on type profile $\theta \in \Theta$.
On the other hand, we consider the formulation with an action set $A_i^{\theta_i}$ depending on the type $\theta_i \in \Theta_i$ for each player $i \in N$.

Given a Bayesian game of the standard formulation, we can formulated it in our formulation by copying each action $\tilde{a}_i \in \tilde{A}_i$ for each type $\theta_i \in \Theta_i$.
Formally, we can define $A_i^{\theta_i} = \{(\theta_i,\tilde{a}_i) \mid \tilde{a}_i \in \tilde{A}_i\}$ for each $i \in N$ and $\theta_i \in \Theta_i$ and $v_i((\theta_1,a_1),\dots,(\theta_n,a_n)) = \tilde{v}_i(\theta; a)$.
This reduction works for any Bayesian game.

If the original game is a valid/basic utility game for each fixed type profile $\theta \in \Theta$, there exists some monotone submodular function $f^\theta \colon 2^{\tilde{A_1} \cup \cdots \cup \tilde{A_n}} \to \bbR_{\ge 0} $ for each $\theta \in \Theta$.
In most applications, we can naturally combine these functions into a single monotone submodular function $f \colon 2^{E} \to \bbR_{\ge 0}$, where $E = \bigcup_{i \in N} \bigcup_{\theta_i \in \Theta_i} A_i^{\theta_i}$ is the set of all actions.
Using this function as the social welfare, we can define it as a Bayesian valid/basic utility game in our formulation.

Here we illustrate the idea of this reduction by using simple resource allocation games.
These generalize Bayesian extensions of congestion games provided by \citet*{Rou15}.

\begin{example}[Resource allocation games with unknown weights]
Let $\tilde{E}$ be the set of resources.
Each player $i \in N$ is associated with possible choices $\tilde{A}_i \subseteq 2^{\tilde{E}}$ and weight $\theta_i \in \bbR_{\ge 0}$.
Each resource $e \in \tilde{E}$ is associated with a non-negative, concave, and monotonically non-decreasing utility function $u_e \colon \bbR_{\ge 0} \to \bbR_{\ge 0}$.
If each player $i \in N$ chooses $P_i \in \tilde{A}_i$, each resource yields payoff $u_e(\sum_{i \in N \colon e \in P_i} \theta_i)$, which is shared by the players selecting $e$ proportionally to their weight, i.e., player $i \in N$ such that $e \in P_i$ receives $\frac{\theta_i}{\sum_{i \in N \colon e \in P_i} \theta_i} u_e(\sum_{i \in N \colon e \in P_i} \theta_i)$ from resource $e$.
Suppose that all players' weights $\theta$ are unknown in advance and jointly generated from a commonly known distribution $\rho$.
Then we can define $E = \left( \bigcup_{i \in N} \Theta_i \right) \times 2^{\tilde{E}}$ and social welfare function $f(X) = \sum_{e \in \tilde{E}} u_e(\sum_{(w,P) \in X \colon e \in P} w)$ for each $X \subseteq E$.
\end{example}

\begin{example}[Resource allocation games with uncertain action sets]
We can naturally deal with uncertain action sets in our formulation.
For example, \citet*{Rou15} considered routing games with uncertain origin-destination pairs.
The set of actions for each player $i \in N$ is the set of all paths that connect two vertices $o_i \in V$ and $d_i \in V$ on a given graph.
If the pair $(o_i, d_i)$ is uncertain and generated from a prior distribution, we can define $\theta_i = (o_i, d_i)$, and let $A_i^{(o_i,d_i)}$ be the set of all paths connecting $o_i$ and $d_i$.
\end{example}

We should note that this reduction does not always apply.
One can easily observe that this reduction does not work if $\{f^\theta\}_{\theta \in \Theta}$ is not consistent.
For example, if the value of $f^\theta(\{\tilde{a}_i\})$ is different depending on $\theta_{-i}$ for some $i \in N$ and $\tilde{a}_i \in \tilde{A}_i$, it is clearly impossible to find a single $f$ consistent with all $\{f^\theta\}_{\theta \in \Theta}$.

Without the assumption of the consistent social welfare function, the SR gap can deteriorate to the trivial lower bound $\Omega(1/n)$, as illustrated by the following simple example.

\begin{example}
There are $n$ players, and the type $\theta_i$ of player $i$ is defined as $\theta_i \equiv ai + b \pmod{n}$, where $a$ and $b$ are random numbers chosen from $\{0,1,\dots,n-1\}$. Each player $i$ chooses an action $a_i$ from $\{0,1,\dots,n-1\}$, and the social welfare function is defined as the number of players who choose $a_i = a$. In this setting, the optimal social welfare is $n$, while the expected social welfare under an optimal strategy profile achieves only $O(1)$.
\end{example}

\subsection{Applications}\label{applications}

We present more detailed descriptions of the examples of Bayesian valid and basic utility games. In these games, it is natural to introduce mediators who aim to coordinate players' decisions, and the social welfare function satisfies submodularity.

\begin{enumerate}
\item In a transportation setting, each commuter (player) $i$ has a private schedule preference (type) $\theta_i$ and must choose a transit option $r$ (e.g., a train or bus line). Each option $r$ has an increasing concave function $u_{r,t} \colon \mathbb{Z} \to \mathbb{R}$ representing the comfort or service quality for commuters of type $t$. Let $N_{r,t}$ denote the set of commuters with type $t$ who choose option $r$. As the number of commuters increases, crowding reduces comfort and efficiency, and each individual in $N_{r,t}$ receives only the marginal improvement in service, $u_{r,t}(|N_{r,t}|) - u_{r,t}(|N_{r,t}|-1)$. The total utility (social welfare) is a monotone submodular function, forming a Bayesian basic utility game. A transportation planner (mediator) recommends routes based on reported schedules to balance the load.

\item In a volunteer task assignment setting, each participant (player) $i$ has a private constraint or preference (type) $\theta_i$, such as time availability or physical ability. The participant selects a task $a_i \in A_i^{\theta_i}$, where $A_i^{\theta_i} \subseteq E$ is the set of feasible tasks for type $\theta_i$. The utility of each participant is their contribution, modeled as the reciprocal of the number of participants assigned to the same task. The total number of completed tasks (social welfare) forms a monotone submodular function. A volunteer coordinator (mediator) collects private information and recommends task assignments to improve efficiency.

\item In a wireless network, each user device (player) $i$ has demand $d_i$ and location $l_i$, which together form the type $\theta_i = (d_i, l_i)$. The device selects a base station $a_i \in E \cap B(l_i, r)$ for communication, where $E$ is the set of all base stations and $B(l_i, r)$ is the ball of radius $r$ centered at $l_i$. The utility of each device is $d_i$ if the total load on the selected base station does not exceed its capacity; otherwise, the excess is proportionally subtracted from each device's demand. The sum of all players’ utilities (social welfare) is a monotone submodular function. A network planner (mediator) collects location and demand information and recommends base stations to balance the load.
\end{enumerate}

\section{Overview of equilibrium concepts for Bayesian games}\label{bce}
This section provides an overview of various equilibrium concepts mentioned in this paper for completeness.
For more detailed descriptions, see \citep*{Forges93} or recent surveys \citep*{Forges06,Fujii23}.

\subsection{Bayes correlated equilibria}
There are multiple natural extensions of correlated equilibria in (complete-information) normal-form games to Bayesian games.

\paragraph{Strategic-form correlated equilibria (SFCEs)}
A Bayesian game can be interpreted as a (complete-information) normal-form game by treating a strategy $s_i \in S_i := \prod_{\theta_i \in \Theta_i} A_i^{\theta_i}$ as a single action.
This interpretation of Bayesian games as a normal-form game is called the strategic form (also referred to as the random-vector model or the normal form).
An SFCE is defined as a correlated equilibrium in this strategic form.
This concept can be implemented by mediators who provide each player with a bundle of action recommendations for all possible types, without knowing the players' actual types.
This recommendation can be likened to providing a comprehensive manual that outlines specific actions to take in every possible scenario.

\begin{definition}[Strategic-form correlated equilibria (SFCEs)]
A distribution $\sigma \in \Delta(S)$ is a strategic-form correlated equilibrium
if
for any $i \in N$ and any $\phiSF \colon S_i \to S_i$, it holds that
\begin{equation*}\label{eq:ic-sf}
\E_{\theta \sim \rho} \left[ \E_{s \sim \sigma} \left[ v_i(s(\theta)) \right] \right]
\ge                              
\E_{\theta \sim \rho} \left[ \E_{s \sim \sigma} \left[ v_i((\phiSF(s_i))(\theta_i),s_{-i}(\theta_{-i})) \right] \right].
\end{equation*}
\end{definition}

\paragraph{Agent-normal-form correlated equilibria (ANFCEs)}
Bayesian games can also be interpreted as normal-form games, an approach known as the agent normal form (also referred to as the posterior-lottery model, the Selten model, or the population interpretation).
In the agent normal form, a single player with different types is hypothetically treated as multiple distinct players.
Specifically, the set of players is represented as $\{(i,\theta_i) \mid i \in N, \theta_i \in \Theta_i\}$.
Once the type profile $\theta \sim \rho$ is realized, only the hypothetical players corresponding to the realized types $(i, \theta_i)$ for each $i \in N$ participate the game and receive payoffs.
An ANFCE is a correlated equilibrium in this agent normal form.
There is no known realistic scenario where ANFCEs naturally serve as an equilibrium concept.

\begin{definition}[Agent-normal-form correlated equilibria (ANFCEs)]
A distribution $\sigma \in \Delta(S)$ is an agent-normal-form correlated equilibrium
if
for any $i \in N$, $\theta_i \in \Theta_i$, and $\phi \colon A_i^{\theta_i} \to A_i^{\theta_i}$, it holds that
\begin{equation*}\label{eq:ic-anf}
\E_{\theta_{-i} \sim \rho|_{\theta_i}} \left[ \E_{s \sim \sigma} \left[ v_i(s(\theta)) \right] \right]
\ge                                        
\E_{\theta_{-i} \sim \rho|_{\theta_i}} \left[ \E_{s \sim \sigma} \left[ v_i(\phi(s_i(\theta_i)), s_{-i}(\theta_{-i})) \right] \right].
\end{equation*}
\end{definition}

\begin{remark}
Similarly, Nash equilibria can be defined in both the strategic form and the agent normal form.
Unlike correlated equilibria, the sets of Nash equilibria in these two normal-form games are identical \citep*{Harsanyi67} and collectively referred to as \textit{Bayes--Nash equilibria}.
\end{remark}

\paragraph{Communication equilibria}
Communication equilibria were first proposed by \citet*{Myerson82} as an extension of Bayesian incentive-compatible mechanism (see \citet*[Chapter 6]{Myerson}).
This equilibrium is realized by bidirectional communication between a mediator and players.
The players tell their types to the mediator via a private communication channel, and then the mediator recommends an action to each player also via a private communication channel.
The mediator's strategy is represented by a type-dependent distribution $\pi \in \prod_{\theta \in \Theta} \Delta(A^\theta)$, that is, given the type profile $\theta \in \Theta$ reported by the players, the mediator recommends an action profile $a \sim \pi(\theta)$.
If any player does not have an incentive to tell an untrue type nor deviate from the recommendation, this type-dependent distribution is called a communication equilibrium.

\begin{definition}[Communication equilibria]
A type-dependent distribution $\pi \in \prod_{\theta \in \Theta} \Delta(A^\theta)$ is a communication equilibrium
if
for any $i \in N$, $\theta_i, \theta'_i \in \Theta_i$, and $\phi \colon A_i^{\theta'_i} \to A_i^{\theta_i}$, it holds that
\begin{equation*}\label{eq:ic-com}
\E_{\theta_{-i} \sim \rho|_{\theta_i}} \left[ \E_{a \sim \pi(\theta)} \left[ v_i(a) \right] \right]
\ge
\E_{\theta_{-i} \sim \rho|_{\theta_i}} \left[ \E_{a \sim \pi(\theta'_i,\theta_{-i})} \left[ v_i(\phi(a_i), a_{-i}) \right] \right].
\end{equation*}
\end{definition}

\begin{remark}
Not all type-dependent distributions $\prod_{\theta \in \Theta} \Delta(A^\theta)$ can be represented by strategy distributions $\Delta(S)$.
\citet*{Fujii23} referred to type-dependent distributions that have an equivalent strategy distribution as \textit{strategy-representable}.
The strategy representability gap represents the ratio of the optimal social welfare achieved by strategy distributions to that achieved by type-dependent distributions.
\end{remark}

\paragraph{Bayesian solutions}
A Bayesian solution is implemented by mediators who know the true type profile or can verify the types reported by players.
The mediator can recommend the (possibly correlated) actions to the players depending on the actual types under the incentive constraints.

\begin{definition}[Bayesian solutions]
A type-dependent distribution $\pi \in \prod_{\theta \in \Theta} \Delta(A^\theta)$ is a Bayesian solution
if
for any $i \in N$, $\theta_i \in \Theta_i$, and $\phi \colon A_i^{\theta_i} \to A_i^{\theta_i}$, it holds that
\begin{equation*}\label{eq:ic-bs}
\E_{\theta_{-i} \sim \rho|_{\theta_i}} \left[ \E_{a \sim \pi(\theta)} \left[ v_i(a) \right] \right]
\ge
\E_{\theta_{-i} \sim \rho|_{\theta_i}} \left[ \E_{a \sim \pi(\theta)} \left[ v_i(\phi(a_i), a_{-i}) \right] \right].
\end{equation*}
\end{definition}

\subsection{Bayes coarse correlated equilibria}

Coarse correlated equilibria in complete-information games are implemented by mediators with somewhat unnatural power to force players to follow the recommendation once they receive them.
In contrast to correlated equilibria, where deviations can depend on the recommendation provided by the mediator, each player must choose to follow the recomendation or deviate to an action without observing the recommendation.

\paragraph{Strategic-form coarse correlated equilibria (SFCCEs) and agent-normal-form coarse correlated equilibria (ANFCCEs)}
As with SFCEs and ANFCEs, we can define coarse correlated equilibria in the strategic form and agent normal form, respectively.

\begin{definition}[Strategic-form coarse correlated equilibria (SFCCEs)]
A distribution $\sigma \in \Delta(S)$ is a strategic-form coarse correlated equilibrium
if
for any $i \in N$ and any $s'_i \in S_i$, it holds that
\begin{equation*}
\E_{\theta \sim \rho} \left[ \E_{s \sim \sigma} \left[ v_i(s(\theta)) \right] \right]
\ge                              
\E_{\theta \sim \rho} \left[ \E_{s \sim \sigma} \left[ v_i(s'_i(\theta_i),s_{-i}(\theta_{-i})) \right] \right].
\end{equation*}
\end{definition}

\begin{definition}[Agent-normal-form coarse correlated equilibria (ANFCCEs)]
A distribution $\sigma \in \Delta(S)$ is an agent-normal-form coarse correlated equilibrium
if
for any $i \in N$, $\theta_i \in \Theta_i$, and $a'_i \in A_i$, it holds that
\begin{equation*}
\E_{\theta_{-i} \sim \rho|_{\theta_i}} \left[ \E_{s \sim \sigma} \left[ v_i(s(\theta)) \right] \right]
\ge                              
\E_{\theta_{-i} \sim \rho|_{\theta_i}} \left[ \E_{s \sim \sigma} \left[ v_i(a'_i,s_{-i}(\theta_{-i})) \right] \right].
\end{equation*}
\end{definition}

\begin{remark}
Whereas SFCEs are a subset of ANFCEs, ANFCCEs are a subset of SFCCEs.
The first inclusion comes from the difference of information used for deviations in SFCEs and ANFCEs;
in SFCEs, each player can deviate to another action depending on the recommended \textit{strategy} (recommendations for all types), while in ANFCEs, depending only the recommended action to the realized type.
The second inclusion comes from the difference of granularity of deviations in SFCCEs and ANFCCEs.
Recall that in CCEs, each player must choose whether to follow the recommendations or stick to a single action without observing recommendations.
If each player decides not to follow the recommendations, the player must stick to a strategy (some action for each type) in SFCCEs.
On the other hand, in ANFCCEs, the player must stick to an action for a single type while following recommendations for the other types.
\end{remark}

\paragraph{Coarse Bayesian solutions}

We can consider CCEs in the setting where the mediator knows the true type profile in advance, which can be called \textit{coarse Bayesian solutions}.
Coarse Bayesian solutions were considered by several PoA/PoS studies \citep*{CKK15,JL23}, and then its difference from other definitions of Bayes coarse correlated equilibria was clarified by \citet*{Fujii23}.
As in SFCCEs and ANFCCEs, we can consider two different versions of coarse Bayesian solutions, each of which allows players to deviate to a strategy or an action for a single type, respectively.

\begin{definition}[{Strategic-form coarse Bayesian solutions}]
A type-dependent distribution $\pi \in \prod_{\theta \in \Theta} \Delta(A^\theta)$ is a \textit{strategic-form coarse Bayesian solution} (SFCBS)
if
for any $i \in N$ and $s_i \in S_i$, it holds that
\begin{equation*}
\E_{\theta \sim \rho} \left[ \E_{a \sim \pi(\theta)} \left[ v_i(a) \right] \right]
\ge
\E_{\theta \sim \rho} \left[ \E_{a \sim \pi(\theta)} \left[ v_i(s_i(\theta_i), a_{-i}) \right] \right].
\end{equation*}
\end{definition}

\begin{definition}[{Agent-normal-form coarse Bayesian solutions}]
A type-dependent distribution $\pi \in \prod_{\theta \in \Theta} \Delta(A^\theta)$ is an \textit{agent-normal-form coarse Bayesian solution} (ANFCBS)
if
for any $i \in N$, $\theta_i \in \Theta_i$, and $a'_i \in A_i$, it holds that
\begin{equation*}
\E_{\theta_{-i} \sim \rho|_{\theta_i}} \left[ \E_{a \sim \pi(\theta)} \left[ v_i(a) \right] \right]
\ge
\E_{\theta_{-i} \sim \rho|_{\theta_i}} \left[ \E_{a \sim \pi(\theta)} \left[ v_i(a'_i, a_{-i}) \right] \right].
\end{equation*}
\end{definition}

SFCBSs and ANFCBSs are defined on type-dependent distributions $\prod_{\theta \in \Theta} \Delta(A^\theta)$, while SFCCEs and ANFCCEs are defined on strategy distributions $\Delta(S)$.
This difference of domains is the only difference between these concepts, and the inclusions indicated in \Cref{figure} hold.

\end{document}